\documentclass[12pt]{article}
\usepackage{natbib}
\usepackage{caption, float}
\usepackage[margin=1in]{geometry}

\usepackage{amsmath,amsthm,amssymb,graphicx,hyperref}
\hypersetup{colorlinks,linkcolor={blue},citecolor={blue},urlcolor={red}}  
\usepackage[linesnumbered,ruled,lined]{algorithm2e}
\usepackage{xcolor}
\usepackage{dsfont}
\usepackage{bm}

\newtheorem{theorem}{Theorem}[section]
\newtheorem{lemma}[theorem]{Lemma}
\newtheorem{proposition}[theorem]{Proposition}

\newenvironment{restatedlemma}[1]{
  \begin{trivlist}
  \item[\hskip \labelsep {\bfseries Lemma~#1.}]\itshape
}{
  \end{trivlist}
}

\newenvironment{restatedproposition}[1]{
  \begin{trivlist}
  \item[\hskip \labelsep {\bfseries Proposition~#1.}]\itshape
}{
  \end{trivlist}
}

\renewcommand*\theassumption{(A\arabic{assumption})}
\newcounter{subassumption}[assumption]

\makeatletter
\renewcommand{\p@subassumption}{\theassumption}
\makeatother

\numberwithin{equation}{section}

\numberwithin{equation}{section}

\newcommand{\set}[1]{\left\{#1\right\}}
\newcommand{\pt}[1]{\left(#1\right)}
\newcommand{\br}[1]{\left[#1\right]}
\newcommand{\norm}[1]{\left\lVert#1\right\rVert}
\newcommand{\N}{\mathbb{N}}
\newcommand{\R}{\mathbb{R}}

\newcommand{\E}{\mathbb{E}} 
\newcommand{\Cov}{\text{Cov}}

\newcommand\given[1][]{\:#1\vert\:}

\DeclareMathOperator*{\argmin}{argmin}

\graphicspath{ {./} }

\begin{document}

\def\spacingset#1{\renewcommand{\baselinestretch}
{#1}\small\normalsize} \spacingset{1}

\title{Sparse Longitudinal Functional Principal Component Analysis for Episodic Ambulatory Behavioral Assessments}

\author{
Nidhi Pai$^{1}$\thanks{Email: pai00032@umn.edu}
\and
Yu Fang$^{2}$
\and
Srijan Sen$^{2,3}$
\and
Zhenke Wu$^{3, 4, 5}$
\and
Erjia Cui$^{1}$
}

\date{
\small
$^{1}$Division of Biostatistics and Health Data Science, University of Minnesota, Minneapolis, MN\\
$^{2}$Michigan Neuroscience Institute, University of Michigan, Ann Arbor, MI\\
$^{3}$Eisenberg Family Depression Center, University of Michigan, Ann Arbor, MI \\
$^{4}$Department of Biostatistics, University of Michigan, Ann Arbor, MI \\
$^{5}$Michigan Institute for Data and AI in Society, University of Michigan, Ann Arbor, MI
}

\maketitle

\begin{abstract}
Accurately monitoring mental fatigue is critical for improving workplace safety and productivity.
A recent study examined unobtrusively collected smartphone typing speed as a potential ambulatory proxy assessment of mental fatigue using data from the Intern Health Study (IHS).
While population-level average typing speed patterns were found to be consistent with validated measures of mental fatigue, how these trajectories vary across participants and days may inform opportune moments for just-in-time interventions and remains an open question.
Treating typing speed trajectories as sparsely observed functional data, we propose a novel sparse longitudinal functional principal component analysis (sparse LFPCA) method for decomposing variability and predicting individual curves.
Specifically, sparse data are accommodated by casting covariance estimation as a structured penalized spline regression problem, enabling simultaneous estimation and smoothing of multiple covariance components while borrowing information across locations in the functional domain.
Simulations show that sparse LFPCA (1) accurately estimates eigenfunctions and generates reasonable predictions for underlying curves, and (2) achieves similar or superior performance compared to existing alternatives.
Our analysis of typing speed data collected from IHS reveals new and interpretable participant- and day-level patterns not captured by previous analyses and can be used to tailor behavioral interventions.
\end{abstract}


\newpage
\baselineskip=24pt

\section{Introduction}
\label{sec:intro}

\subsection{Motivation}
\label{subsec:intro_motivation}

Advances in ambulatory assessment technologies hold tremendous promise to unobtrusively probe the social, behavioral, physiological contexts in which individuals live and experience conditions that require clinical attention. 
Recognizing this potential, modern health and biomedical studies increasingly adopt these novel technologies to collect data repeatedly over time on the same subjects for clinical endpoint prediction, association studies, and constructing tailoring variables for just-in-time interventions. 
For example, the Apple SensorKit \citep{apple_sensorkit, langholm2023sensorkit} platform passively collects typing speed and accuracy every day at irregular and sparse times that are associated with typing activities. 
This information was collected in our motivating Intern Health Study (IHS), an ongoing large annual cohort study examining mental health and stress in first-year physician residents. \cite{fang2026sensorkit_fatigue} used two months of data from the start of the internship to investigate smartphone typing performance in the IHS as a potential indicator for monitoring mental fatigue, an important task for improving workplace safety. 
In addition to typing speed metrics, behavioral and physiological measures, including sleep patterns, were extracted from fitness trackers worn by the participants.
For each typing session (defined by a short period lumping multiple temporally adjacent typing events across all apps), time awake was defined as the elapsed time between the typing initiation and the end of the most recent sleep episode. 
As a result, for each participant on each day, typing speed was collected at several discrete typing sessions (one observation per typing session) and treated as function of time awake. See \citet{fang2026sensorkit_fatigue} for the exact definitions of the typing speed.

\begin{figure}
    \centering
    \includegraphics[width=.75\linewidth]{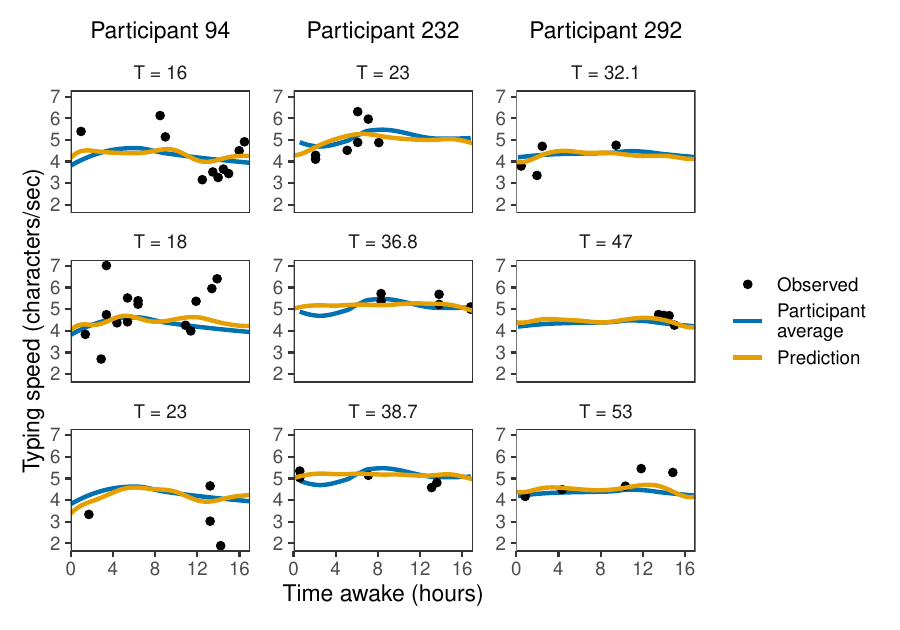}
    \caption{Selected participants and days from the IHS typing speed data. Columns correspond to different participants, and three days of data are shown for each participant. The x-axis is time since waking, in hours, and the y-axis is typing speed, in characters per second. The time for each day (denoted by $T$) is the elapsed time (in days) from 6:00 a.m. on the start day of the study (June 29) to the wake-time of the most recent sleep episode. For example, $T = 23$ implies the participant woke up around 6:00 a.m. on July 22, whereas $T = 36.8$ indicates the participant woke up around 1:12 a.m. on August 5.
    Black points represent typing sessions, and blue curves represent the participant's smoothed average typing speed, pooled over all available days for the participant.
    Gold lines are predictions generated by our proposed methods, as described in Sections \ref{sec:methods} and \ref{sec:application}.
    }
    \label{fig:sk_figure}
\end{figure}

To illustrate, Figure~\ref{fig:sk_figure} shows data from three participants over three selected days per participant, with participants shown in separate columns.
The x-axis is time awake in hours, and the y-axis shows typing speed in characters per second.
Black points represent typing sessions, and blue curves represent the participant's smoothed average typing speed, pooled over all available days for the participant.
The number of typing sessions per day varies between 1 and 36, with a median of 3 sessions per day.
The data is sparse, as most days contain very few observations, and irregular, since sampling times vary within and across days.
In addition, each participant has between 1 and 74 days of data (median 28.1 days), resulting in sparse trajectories (a few observations per day) collected longitudinally (multiple days per participant).
This is a more complicated setting than traditional longitudinal studies because a sparse trajectory rather than a scalar measurement is observed per day per participant.
Using this data, \cite{fang2026sensorkit_fatigue} investigated population-level patterns in typing speed and evaluated its utility as a proxy for mental fatigue.
While they found that average typing speed mirrors patterns in lab-based measures of mental fatigue in the literature, an important scientific question remains: how do these trajectories vary across participants and across days within participants?
Periods of extreme within-participant variation may indicate opportune moments for digital health interventions, e.g., push notifications suggesting a break when typing speed drops significantly and the participant is available to be notified.

\subsection{Challenges Faced by the Existing Methods}
\label{subsec:intro_existing}

To answer this question, a natural approach is to model daily typing speed trajectories as sparse functional observations and analyze them using tools from functional data analysis (FDA; \cite{ramsay2005fda, kokoszka2017introduction, crainiceanu2024fda_book, jiang2025tutorial}).
Functional principal components analysis (FPCA; \cite{yao2005fpca, ramsay2005fda, xiao2016denseface}) is a widely used technique in FDA, first designed to investigate dominant patterns of variation in independent, single-level functional data.
Multilevel FPCA (MFPCA) extends FPCA to hierarchical settings by incorporating a functional random intercept \citep{di2009mfpca, cui2023fmfpca}.
For longitudinal functional observations, as in the Intern Health Study, \cite{greven2010lfpca} proposed longitudinal FPCA (LFPCA) which further extends MFPCA by including a functional random slope to model variation across visits.
A variety of longitudinal functional models have since been developed to address inference \citep{park2015longitudinal, scheipl2015famm, shamshoian2022bayesian, li2022fixed}, large, high-dimensional datasets \citep{cui2022fui, zipunnikov2014longitudinal, loewinger2025fast}, asynchronous variables \citep{li2022regression}, and skewed distributions \citep{alam2024modeling}.
For sparse or irregular functional data, a parallel line of research has focused on FPCA  \citep{james2000principal, sparse_face}, with extensions to the multilevel \citep{di2014sparse_mfpca} and multivariate \citep{li2020multi_sface, ghosh2025gof} settings.
Some work takes a Bayesian approach \citep{sartini2025bayesian, ye2024functional}, while others focus on addressing non-Gaussian data \citep{zhong2022robust}.

Despite substantial progress in both directions, existing FPCA methods are limited in their ability to simultaneously account for both longitudinal and sparse functional data structures.
The LFPCA framework of \cite{greven2010lfpca} was developed for dense functional data observed on a common time grid. It estimates the covariance at each pair of time points separately, without borrowing information between pairs of time points.
In sparse or irregular designs, the limited number of observations at any given pair of time points leads to unstable, high-variance, or even infeasible estimates.
This problem was considered by \cite{cederbaum2016sparse_flmm}, who introduced functional linear mixed models (FLMM) for irregular or sparse data.
However, their method applies only to a simplified crossed random effect structure and is computationally intensive. See Supplementary Section~\ref{sec:s_review} for details on limitations of existing methods.

\subsection{Contributions}
\label{subsec:intro_contributions}

To address these limitations, we propose a novel sparse LFPCA framework. 
Our primary contribution is to cast covariance estimation as a structured penalized spline regression problem, which enables estimating and smoothing \textit{multiple} covariance components simultaneously.
Although \cite{sparse_face} proposed a fast covariance estimation approach for sparse functional data, their method applies only to single-level observations.
To solve critical new estimation challenges in the longitudinal functional setting, as in the motivating SensorKit data, we (1) propose a refined weight matrix to account for longitudinal correlations, (2) introduce a novel block-structured smoothness penalty, and (3) develop a scalable smoothing parameter selection strategy that regularizes separate covariance functions simultaneously while allowing each to have a different degree of smoothness.
In addition, we provide ready-to-use R software, \texttt{lfpca.sparse()}, implementing our method.
To our knowledge, this is the first time FPCA has been extended to a sparse longitudinal setting.

The remainder of this article is organized as follows.
Section~\ref{sec:methods} reviews the LFPCA model and introduces our proposed sparse LFPCA framework.
Section~\ref{sec:simulation} contains simulation studies of the proposed method, including a comparison with MFPCA as a special case.
Section~\ref{sec:application} presents an application of our model to the Intern Health Study data. We conclude with a discussion in Section~\ref{sec:discussion}.

\section{Methods}
\label{sec:methods}

\subsection{LFPCA Model}
\label{subsec:model}

We begin by briefly reviewing the LFPCA model.
For subject $i$, $i = 1, ..., I$ at visit $j$, $j = 1, ..., J_i$, suppose we observe $Y_{ij}(s)$ at locations $\set{s_{ijk}}_{k = 1, ..., m_{ij}} \subset \mathcal{S}$, where $\mathcal{S} = [0, 1]$ is the functional domain and $m_{ij}$ is the number of observations for subject $i$ at visit $j$.
In the dense data setting considered by \cite{greven2010lfpca}, the locations $s_{ijk}$ belong to a regular grid $\mathcal{S}^* = \{s_1, s_2, \dots, s_L\}$ shared across subjects and visits.
Let $T_{ij}$ be the time of visit $j$ for subject $i$.
The longitudinal functional model is
\begin{equation}
    Y_{ij}(s) = \underbrace{\mu(s, T_{ij})}_{\sf population~fixed~effects} + \underbrace{\bm R_{ij}^\top \bm Z_i(s)}_{\sf subject-level~random ~
    effects}  + \underbrace{W_{ij}(s)}_{\sf visit-level~random~effects} + \underbrace{\epsilon_{ij}(s)}_{\sf residual~error},
    \label{eq:model}
\end{equation}
where $\mu(s, T_{ij})$ denotes the fixed effect, 
$\bm Z_i(s) = (Z_{i, 0}(s), Z_{i, 1}(s))^\top$
is a vector of functional subject-level random effects with corresponding design matrix $\bm R_{ij} = (1, T_{ij})^\top$,
$W_{ij}(s)$ is the visit-specific deviation, and $\epsilon_{ij}(s)$ is white noise with zero mean and variance $\sigma^2$.
Here, $\bm Z_i(s)$, $W_{ij}(s)$, and $\epsilon_{ij}(s)$ are centered, mutually uncorrelated stochastic processes.
Let $K_0(s,u) = \Cov\{Z_{i, 0}(s), Z_{i, 0}(u)\}$ and $K_1(s,u) = \Cov\{Z_{i, 1}(s),  Z_{i, 1}(u)\}$ be the auto-covariance functions of $Z_{i, 0}(s)$ and $Z_{i, 1}(s)$, respectively.
Let $K_{01}(s,u) = \Cov\{ Z_{i, 0}(s),  Z_{i, 1}(u)\}$ be the cross-covariance function of $ Z_{i, 0}(s)$ and $ Z_{i, 1}(u)$.
Let $K_{\bm Z}(s,u)$ be the covariance of $\bm Z_i(s)$ such that
$K_{\bm Z}(s,u) = \begin{pmatrix}
    K_0(s,u) & K_{01}(s,u) \\
    K_{01}(u,s) & K_1(s,u)
\end{pmatrix}$.
The functions $K_0(s,u)$ and $K_1(s,u)$ are symmetric in $s$ and $u$, whereas $K_{01}(s,u)$ is generally not symmetric.
Let $K_W(s,u) = \Cov\{W_{ij}(s), W_{ij}(u)\}$ be the covariance function of $W_{ij}(s)$.

LFPCA identifies the main directions of variation and provides dimensionality reduction in longitudinal functional data.
By Mercer's Theorem \citep{mercer1909}, the covariance function $K_{\bm Z}(s,u)$ can be expanded as
$K_{\bm Z}(s,u) = \sum_{n=1}^\infty \lambda_n^Z \bm \phi_n(s) \bm \phi^\top_n(u)$,
where $\lambda_1^Z \geq \lambda_2^Z \geq \cdots \geq 0$ are the eigenvalues,
$\bm \phi_n(s) = (\phi_n^0(s), \phi_n^1(s))^\top$
are the corresponding eigenfunctions, and the set $\set{\bm \phi_n \given n \in \N}$ forms an orthonormal basis with respect to the additive scalar product
\[\langle \pt{\phi^0_{n_1},\phi^1_{n_1}}, \pt{\phi^0_{n_2},\phi^1_{n_2}}  \rangle = \int_\mathcal{S} \phi_{n_1}^0(s) \phi_{n_2}^0(s) ds + \int_\mathcal{S} \phi_{n_1}^1(s) \phi_{n_2}^1(s) ds.\]
Similarly, the eigendecomposition of $K_W(s,u)$ is $K_W(s,u) = \sum_{n = 1}^\infty \lambda_n^W \psi_n(s) \psi_n(u)$, where $\lambda_1^W \geq \lambda_2^W \geq ... \geq 0$ are the eigenvalues and $\psi_n(s)$ are the corresponding eigenfunctions. 
The Kosambi–Karhunen–Loève (KKL; \cite{kosambi1943, karhunen1947, loeve1948functions}) expansions ensure that $\bm Z_i(s) = \sum_{n_1 = 1}^{\infty} \xi_{i n_1} \bm \phi_{n_1}(s) $ and
$W_{ij}(s) = \sum_{n_2 = 1}^{\infty} \zeta_{i j n_2} \psi_{n_2}(s)$,
where $\xi_{i n_1} \sim N(0, \lambda_{n_1}^Z)$ and $\zeta_{ij n_2} \sim N(0, \lambda_{n_2}^{W})$ are mutually uncorrelated scores which are independent of the error. 
One key assumption in LFPCA is that the infinite expansions can be approximated by the first few principal components, i.e., $\bm Z_i(s) \approx \sum_{n_1 = 1}^{N_Z} \xi_{i n_1} \bm \phi_{n_1}(s)$ and $\bm W_{ij}(s) \approx \sum_{n_2 = 1}^{N_W} \zeta_{i j n_2} \psi_{n_2}(s)$, where $N_Z$ and $N_W$ are estimated by techniques such as leave-one-curve-out cross-validation \citep{rice1991lococv}, AIC-type criteria \citep{yao2005fpca}, restricted likelihood ratio tests \citep{staicu2010rlrt1, crainiceanu2009rlrt2}, or proportion of variance explained \citep{greven2010lfpca}.
Then, the model \eqref{eq:model} becomes
\begin{align}
    Y_{ij}(s) &\approx \mu(s, T_{ij}) 
    + \sum_{n_1 = 1}^{N_Z}  \xi_{i n_1} \bm R_{ij}^\top  \bm \phi_{n_1}(s) 
    + \sum_{n_2 = 1}^{N_W} \zeta_{i j n_2} \psi_{n_2}(s) + \epsilon_{ij}(s).
    \label{eq:kkl_model}
\end{align}
The decompositions reduce the longitudinal functional data to two finite sets of scores, $\set{\xi_{in_1} \given i = 1, ..., I; n_1 = 1, ..., N_Z}$ and $\set{\zeta_{ijn_2} \given i = 1, ..., I; j = 1, ..., J_i; n_2 = 1, ..., N_W}$, yielding a low-dimensional representation. 
The scores can be used to predict underlying curves or for downstream tasks such as regression or clustering.

\subsection{Sparse LFPCA Framework}
\label{subsec:methods_outline}

To overcome the limitations of existing methods, we propose an LFPCA estimation framework for sparse functional data. The main steps, summarized in Algorithm~\ref{alg:methods}, are to estimate the mean function, estimate the between-subject covariance functions, estimate the within-subject covariance function, and predict subject- and visit-level scores.
Note that in our framework, MFPCA can be viewed as a special case without the longitudinal component; the necessary modifications are described in the Supplementary Materials, Section \ref{sec:s_mfpca}.

\begin{algorithm}
\caption{Sparse LFPCA}\label{alg:methods}
\begin{enumerate}
    \item Estimate the mean function $\mu(s, T_{ij})$ under the working independence model $Y_{ij}(s) = \mu(s, T_{ij}) + \epsilon_{ij}(s)$.
    \item Estimate $K_0(s,u)$, $K_{01}(s,u)$, and $K_1(s,u)$ using cross-products from different visits $\set{\widehat C_{i j_1 j_2 k_1 k_2} \given i = 1, ..., I; j_1 = 1, ..., J_i; j_2 \not = j_1; k_1 = 1, ..., m_{ij_1}; k_2 = 1, ..., m_{i j_2}}$. Obtain eigenvectors $\bm \phi_{n_1}(t)$ and eigenvalues $\lambda_{n_1}^Z$.
    \item Estimate $K_W(s,u)$ and $\sigma^2$ using adjusted cross-products from the same visit 
    $\set{\widehat A_{i j j k_1 k_2} \given i = 1, ..., I; j = 1, ..., J_i; k_1, k_2 = 1, ..., m_{ij}}$. Obtain eigenvectors $\psi_{n_2}(t)$ and eigenvalues $\lambda_{n_2}^W$.
    \item Obtain scores $\xi_{i n_1}$ and $\zeta_{i j n_2}$ using mixed model equations. Substitute the mean function and eigenfunctions into Equation~\eqref{eq:kkl_model} to obtain predictions for $Y_{ij}(s)$.
\end{enumerate}
\end{algorithm}

\subsection{Fixed Effects Estimation}

The first step is to obtain an estimate $\widehat \mu(s, T_{ij})$ of the mean surface under a working independence assumption.
In the most general case, $\mu(s, T_{ij})$ can be estimated as a smooth bivariate function of $s$ and $T_{ij}$, which is appropriate if the collection of visit times $T_{ij}$ over all subjects is dense.
Otherwise, simpler fixed effect structures can be used such as $\mu(s, T_{ij}) = \mu(s) + T_{ij} \beta$.
We estimate $\mu(s, T_{ij})$ using tensor product P-splines in a generalized additive framework \citep{eilers_marx1996_psplines}.
Choosing between smoothing methods is not our focus; see \cite{krivobokova2007correlatederrors, wood2017gam} for a more detailed discussion.
After estimating the fixed effects, let $\hat r_{ijk} = Y_{ij}(s_{ijk}) - \hat \mu(s_{ijk}, T_{ij})$ be the residuals from substituting the estimated mean surface. 
We will use the hat notation to denote variables where the mean function $\mu(s, T_{ij})$ is replaced by its estimate $\hat \mu(s, T_{ij})$.

\subsection{Between-Subject Covariance Estimation}
\label{subsec:kb}

\subsubsection{Tensor product spline expansion}
\label{subsubsec:kb_model}

Let $C_{i j_1 j_2 k_1 k_2} = r_{ij_1 k_1} r_{i j_2 k_2}$ be cross-products of the residuals, where $j_1, j_2 \in \set{1, ..., J_i}$ index visits and $k_1, k_2 \in \set{1, ..., m_{ij}}$ index observations for each subject; these cross-products underpin the covariance estimation.
Based on the model in Equation~\eqref{eq:model}, the expected value of $C_{i j_1 j_2 k_1 k_2}$ is
\begin{equation}
    \begin{aligned}
    \E \br{C_{i j_1 j_2 k_1 k_2}} 
    &= K_0(s_{i j_1 k_1}, s_{i j_2 k_2}) + T_{ij_2} K_{01}(s_{i j_1 k_1}, s_{i j_2 k_2}) + T_{ij_1} K_{01}(s_{i j_2 k_2}, s_{i j_1 k_1}) \\
    &\quad + T_{ij_1} T_{i j_2} K_{1}(s_{i j_1 k_1}, s_{i j_2 k_2}) + \br{K_W(s_{i j_1 k_1}, s_{i j_2 k_2}) + \sigma^2 \delta_{k_1  k_2} } \delta_{j_1 j_2},
    \end{aligned}
        \label{eq:expected_value}
\end{equation}
where $\delta$ is the Kronecker delta. In the LFPCA method proposed by \cite{greven2010lfpca}, this model is fit as a linear regression, where the parameters to be estimated are the covariance functions at each pair of locations, and the raw covariance estimates are subsequently smoothed.
However, with sparse data, the number of cross-products is often limited, making it difficult to accurately estimate the many unknown parameters in the linear regression step. 
See Section \ref{sec:s_review} for more details on the limitations of existing methods for sparse data.

To improve covariance estimation in the sparse data setting, we propose incorporating smoothing to borrow information across location pairs.
Specifically, we model the between-subject covariance functions $K_0(s,u)$, $K_{01}(s,u)$, and $K_1(s,u)$ with tensor product splines
\begin{equation}
\begin{aligned}
    H_0(s,u) &= \sum_{1 \leq \kappa, \ell \leq c} \theta_{\kappa \ell}^{0} B_{\kappa}(s) B_{\ell}(u), \qquad 
    H_{01}(s,u) = \sum_{1 \leq \kappa, \ell \leq c} \theta_{\kappa \ell}^{01} B_{\kappa}(s) B_{\ell}(u), \\
    H_1(s,u) &= \sum_{1 \leq \kappa, \ell \leq c} \theta_{\kappa \ell}^{1} B_{\kappa}(s) B_{\ell}(u),
\end{aligned}
\label{eq:kb_basis}
\end{equation}
respectively. 
Here, $\set{B_1(\cdot), ..., B_c(\cdot)}$ is a collection of B-spline basis functions in $\mathcal{S}$, $\set{\theta^0_{\kappa \ell}, \theta^{01}_{\kappa \ell}, \theta^1_{\kappa \ell}}_{1 \leq \kappa, \ell \leq c}$ are basis function coefficients, and $c$ is the rank of the B-spline basis; knots are equally spaced.
This tensor product spline construction provides a structured and flexible representation and has been used in covariance estimation for single-level sparse data.
However, to our knowledge, it has not been extended to the longitudinal sparse settings.
For simplicity, we assume that the basis functions are the same for the three covariance surfaces. 
For the autocovariance functions $K_0(s,t)$ and $K_1(s,t)$, we enforce that
$\theta_{\kappa \ell} = \theta_{\ell \kappa}$ for $1 \leq \kappa, \ell \leq c$,
which forces $H_0(s,t) = H_0(t,s)$ and $H_1(s,t) = H_1(t,s)$.

For between-subject covariance estimation, consider the set $\set{\widehat C_{i j_1 j_2 k_1 k_2}}$ for all $i = 1, 2, \dots, I$, $j_1 \not = j_2$, and $1 \leq k_1 \leq m_{i j_1}$, $1 \leq k_2 \leq m_{i j_2}$, i.e., products of residuals from the same subject and different visits. 
Subjects with only one visit are not included here.
Based on Equation \eqref{eq:expected_value}, for $j_1 \not = j_2$, we have
$
    \E \br{C_{i j_1 j_2 k_1 k_2}} 
    = K_0(s_{ij_1 k_1}, s_{ij_2 k_2}) + T_{i j_2} K_{01}(s_{ij_1 k_1}, s_{ij_2 k_2}) + T_{i j_1} K_{01}(s_{ij_2 k_2}, s_{ij_1 k_1}) 
     + T_{i j_1} T_{i j_2} K_1(s_{ij_1 k_1}, s_{ij_2 k_2}).
$
Substituting the covariance functions with their basis expansions in Equation~\eqref{eq:kb_basis} gives
\begin{equation}
    \begin{aligned}
    \E \br{C_{i j_1 j_2 k_1 k_2}} 
    &= \sum_{1 \leq \kappa, 1 \leq \ell \leq c} \theta_{\kappa \ell}^{0} B_{\kappa}(s_{ij_1 k_1}) B_{\ell}(s_{ij_2 k_2}) +
    T_{i j_2}\sum_{1 \leq \kappa, 1 \leq \ell \leq c} \theta_{\kappa \ell}^{01} B_{\kappa}(s_{ij_1 k_1}) B_{\ell}(s_{ij_2 k_2}) \\
    &\quad + T_{i j_1}\sum_{1 \leq \kappa, 1 \leq \ell \leq c} \theta_{\kappa \ell}^{01} B_{\kappa}(s_{ij_2 k_2}) B_{\ell}(s_{ij_1 k_1}) 
    + T_{i j_1} T_{i j_2} \sum_{1 \leq \kappa, 1 \leq \ell \leq c} \theta_{\kappa \ell}^{1} B_{\kappa}(s_{ij_1 k_1}) B_{\ell}(s_{ij_2 k_2})  \\
    &= \sum_{1 \leq \kappa, 1 \leq \ell \leq c} \theta_{\kappa \ell}^{0} B_{\kappa}(s_{ij_1 k_1}) B_{\ell}(s_{ij_2 k_2}) \\
    &\quad + \sum_{1 \leq \kappa, 1 \leq \ell \leq c} \theta_{\kappa \ell}^{01} \br{T_{i j_2} B_{\kappa}(s_{ij_1 k_1}) B_{\ell}(s_{ij_2 k_2}) + T_{i j_1} B_{\kappa}(s_{ij_2 k_2}) B_{\ell}(s_{ij_1 k_1}) } \\
    &\quad + \sum_{1 \leq \kappa, 1 \leq \ell \leq c} \theta_{\kappa \ell}^{1} T_{i j_1} T_{i j_2} B_{\kappa}(s_{ij_1 k_1}) B_{\ell}(s_{ij_2 k_2}).
    \end{aligned}
    \label{eq:kb_reg_model}
\end{equation}
While Equation~\eqref{eq:kb_reg_model} looks more complex than Equation~\eqref{eq:expected_value}, the key observation is that we can still write it as a regression model $\widehat{\bm C} = \bm X \bm \alpha$.
Here, the outcome $\widehat{\bm C}$ contains $\widehat C_{i j_1 j_2 k_1 k_2}$ for all $i = 1, 2, \dots, I$, $j_1 \not = j_2$.
To define $\bm \alpha$, first let $\bm{\Theta}_0 = (\theta^0_{k \ell})_{1 \leq \kappa \leq c, 1 \leq \ell \leq c} \in \R^{c \times c}$ be the coefficient matrix for $K_0(s,u)$, and define $\bm \Theta_{01}$ and $\bm \Theta_{1}$ analogously. 
Let $\text{vech}(\cdot)$ denote the operator that stacks the columns of the lower triangle of a matrix into a vector and $\text{vec}(\cdot)$ denote the operator that stacks all columns of a matrix into a vector.
Because $K_0(s,u)$ and $K_1(s, u)$ are symmetric, we need only estimate $\bm \theta_0 = \text{vech}\bm \Theta_0$ and $\bm \theta_1 = \text{vech}\bm \Theta_1$, whereas $K_{01}(s,u)$ is not symmetric, so we estimate $\bm \theta_{01} = \text{vec} \bm \Theta_{01}$.
The vector of regression coefficients is then defined as $\bm \alpha = (\bm \theta_0^\top, \bm \theta_{01}^\top, \bm \theta_1^\top)^\top$.

The design matrix $\bm X$ has a block structure $\bm X = [\bm X_0, \bm X_{01}, \bm X_1]$, where the columns of the three blocks correspond to the entries of $\bm \theta_0$, $\bm \theta_{01}$, and $\bm \theta_1$.
Let $\bm b(s) = \set{B_1(s), ..., B_c(s)}^\top$ be a vector of the basis functions evaluated at time $s$.
To see how $\bm X$ is constructed, first consider $K_0(s,t)$. Note that $H_0(s,u) = (\bm b(u) \otimes \bm b(s))^\top \bm G_c \bm \theta_0$, where $\otimes$ is the Kronecker product and $\bm G_c \in \R^{c^2 \times \frac{c(c+1)}{2}}$ is the duplication matrix (\cite{seber2008_duplication_matrix}, p. 246) such that $\text{vec} \bm \Theta_0 = \bm G_c \bm \theta_0$. 
Using this notation, the block of the design matrix $\bm X$ for estimating $\bm \theta_0$ can be written as $\bm X_0 = \bm B \bm G_c$, where $\bm B$ is a complex matrix consisting of $\bm b(s)$ evaluated at time points corresponding to $\widehat{\bm C}$ and appropriately concatenated.
Blocks $\bm X_{01}$ and $\bm X_1$ are defined similarly, but in more complicated forms involving visit times in addition to basis functions.
The details of constructing $\widehat{\bm C}$, $\bm B$, $\bm X$, and $\bm \alpha$ are given in the Supplementary Materials, Section \ref{subsec:s_kb_matrix}.

\subsubsection{Weighted least squares}
\label{subsubsec:kb_weights}

To increase estimation efficiency, we fit the model with weighted least squares (WLS), where the weight matrix $\bm W$ is specified as the inverse of $\Cov(\bm C)$, as proposed in \cite{sparse_face}.
However, the derivation of $\Cov(\bm C)$ for sparse LFPCA is not straightforward because the longitudinal setting introduces additional model components.
Proposition~\ref{prop:weight} describes how to compute $\Cov(\bm C)$ in this case.

\begin{proposition}
\label{prop:weight}
Define $\bm{M}_{i j_1 j_2 k_1 k_2} = (
    K_0(s_{i j_1 k_1}, s_{i j_2 k_2}),
    T_{i j_2} K_{01}(s_{i j_1 k_1}, s_{i j_2 k_2}),$  \\
    $T_{i j_1} K_{01}(s_{i j_2 k_2}, s_{i j_1 k_1}),
    T_{i j_1} T_{i j_2} K_1(s_{i j_1 k_1}, s_{i j_2 k_2}),
    K_W(s_{i j_1 k_1}, s_{i j_2 k_2}) \delta_{j_1 j_2},
    \sigma^2 \delta_{j_1 j_2} \delta_{k_1 k_2}
)^\top$. Then
\[\Cov(C_{i j_1 j_2 k_1 k_2}, C_{i j_3 j_4 k_3 k_4}) = \bm{1}^\top \pt{\bm{M}_{i j_1 j_3 k_1 k_3} \otimes \bm{M}_{i j_2 j_4 k_2 k_4} + \bm{M}_{i j_1 j_4 k_1 k_4} \otimes \bm{M}_{i j_2 j_3 k_2 k_3}}.
\]
\end{proposition}

The proof of Proposition~\ref{prop:weight} is in the Supplementary Materials, Section \ref{sec:s_weight}.
For the between-subject covariance function estimation, we use $\widehat{C}_{i j_1 j_2 k_1 k_2}$ where $j_1 \not = j_2$. Thus, only the first four entries of $\bm M_{i j_1 j_2 k_1 k_2}$ are nonzero. 
Because $\Cov (\bm C_i)$ may be singular, we let $0 < \beta < 1$ and further define
$\bm W_i^{-1} = (1 - \beta) \Cov (\bm C_i) + \beta \text{diag}\pt{\text{diag}(\Cov (\bm C_i))}$, 
which ensures the existence and numerical stability of $\bm W_i$.
The weight matrix is then constructed as $\bm W = \text{blockdiag}(\bm W_1, ..., \bm W_I)$.
Because $\bm W$ and the analogous weight matrix for the within-subject covariance, $\widetilde{\bm W}$, depend on the covariance functions and $\sigma^2$, a two-stage estimation procedure can be used.
First, set $\bm W = \bm I$ and $\widetilde{\bm W} = \bm I$ and estimate the covariance functions as outlined in Sections \ref{subsubsec:kb_model} and \ref{subsec:kw} using (unweighted) least squares. Second, calculate the weight matrices and repeat the covariance function estimation using WLS.

\subsubsection{Penalized estimation}
\label{subsubsec:kb_estimation}

Because the number of basis functions is relatively large, we propose adding a penalty term to avoid overfitting.
The penalties on $\bm \Theta_0$, $\bm \Theta_{01}$, and $\bm \Theta_1$ are $\lambda_0 ||\bm \Theta_{0} \bm{D}||^2_F$, $\lambda_{01} ||\bm \Theta_{01} \bm{D}||^2_F$, and $\lambda_1 ||\bm \Theta_1 \bm{D}||^2_F$, respectively, where $\lambda_0$, $\lambda_{01}$, and $\lambda_1$ are distinct smoothing parameters, $\bm{D} \in \R^{c \times (c -2)}$ is a second-order difference matrix, and $||\cdot||_F$ is the Frobenius norm.
To unify the penalties for multiple covariance components within a single least squares objective function, we introduce a novel block-structured penalty.
Specifically, define
\begin{align*}
    \bm Q_0 &= \begin{pmatrix}
    \bm G_c^\top (\bm I_c \otimes \bm D \bm D^\top) \bm G_c & \bm 0 & \bm 0 \\
    \bm 0 & \bm 0  & \bm 0 \\
    \bm 0 & \bm 0  & \bm 0 
    \end{pmatrix}, \quad 
    \bm Q_{01} = \begin{pmatrix}
    \bm 0 & \bm 0  & \bm 0  \\
    \bm 0 & \bm I_c \otimes \bm D \bm D^\top  & \bm 0 \\
    \bm 0 & \bm 0  & \bm 0 
    \end{pmatrix}, \\
    \bm Q_{1} &= \begin{pmatrix}
    \bm 0 & \bm 0  & \bm 0  \\
    \bm 0 & \bm 0  & \bm 0  \\
    \bm 0 & \bm 0 & \bm G_c^\top (\bm I_c \otimes \bm D \bm D^\top) \bm G_c
    \end{pmatrix}.
\end{align*}
The penalty is then $\sum_{l \in \mathcal{L}}  \lambda_l \bm \alpha^\top \bm Q_l  \bm \alpha$, where $\mathcal{L} = \set{1, 01, 1}$ are labels corresponding to the three between-subject covariance functions.
Using the duplication matrix $\bm G_c$ ensures that smoothness penalties are correctly imposed on the symmetric coefficient matrices $\bm \Theta_0$ and $\bm \Theta_1$, even though the estimation targets the vectors of unique elements $\bm \theta_0$ and $\bm \theta_1$.
Additionally, the block-structured penalty assigns each covariance component its own smoothing parameter, allowing different components to be smoothed to varying degrees.
We will discuss the selection of smoothing parameters $\lambda_l$ in Section~\ref{subsubsec:kb_sp_selection}.

With the introduced penalty, the resulting penalized WLS objective function is
\[\widehat{\bm \alpha} = \argmin_{\bm \alpha} \br{ (\widehat{\bm C} - \bm{X} \bm \alpha)^\top \bm W (\widehat{\bm C}  - \bm{X} \bm \alpha) + \sum_{l \in \mathcal{L}}  \lambda_l \bm \alpha^\top \bm Q_l  \bm \alpha },\]
and an explicit form for $\hat{\bm{\alpha}}$ is
\[\widehat{ \bm \alpha} =  \pt{\bm{X}^\top \bm W \bm X +  \sum_{l \in \mathcal{L}} \lambda_l \bm Q_l  }^{-1} \pt{\bm X^\top \bm W \widehat{\bm C} }.\]

In practice, we evaluate the covariance functions, eigenfunctions, and predictions on a grid $\mathcal{S}^*$ of $L$ locations in $\mathcal{S}$.
Let $\bm B^* = (\bm b(s_1), ..., \bm b(s_L))^\top \in \R^{L \times c}$ be a matrix of the basis functions evaluated at $\mathcal{S}^*$.
Given coefficient estimates $\widehat{\bm \Theta}_0$, $\widehat{\bm \Theta}_{01}$, and $\widehat{\bm \Theta}_1$ from $\widehat{\bm \alpha}$, the estimated covariance functions evaluated on $\mathcal{S}^*$ are 
$\bm{\widehat{K}}^R_0 = \bm B^* \widehat{\bm \Theta}_0 (\bm B^*)^\top$, 
$\bm{\widehat{K}}^R_{01} = \bm B^* \widehat{\bm \Theta}_{01} (\bm B^*)^\top$, and $\bm{\widehat{K}}^R_1 = \bm B^* \widehat{\bm \Theta}_1 (\bm  B^*)^\top$.
Here, the $R$ superscript stands for ``raw"; the estimates are trimmed in Section~\ref{subsubsec:eigen_kb} to ensure the matrices are positive semi-definite.

\subsubsection{Selection of smoothing parameters}
\label{subsubsec:kb_sp_selection}

We propose selecting smoothing parameters $\lambda_0$, $\lambda_{01}$, and $\lambda_1$ by leave-one-subject-out cross validation, since it accounts for both within-visit and between-visit correlation, unlike leave-one-observation-out and leave-one-visit-out methods.
The leave-one-subject-out cross validated error is $\text{iCV} = \sum_{i =1}^I \norm{\widehat{\bm C}^{[i]}_i - \widehat{\bm C} }^2$, where $\widehat{\bm C}^{[i]}_i$ is the prediction of $\widehat{\bm C}_i$ generated by fitting the model without data from the $i$th subject.
Specifically, we derive an expression for the cross-validated error for longitudinal functional data that is much faster to compute, adapting the approach of \cite{sparse_face}.
First, let $\bm S = \bm X(\bm X^\top \bm W \bm X + \bm Q )^{-1} \bm X^\top \bm W$ be the smoother matrix, where $\bm Q = \sum_{l \in \mathcal{L}} \lambda_l \bm Q_l $.
We then have Lemma~\ref{lemma:kb_smoother}:
\begin{lemma}
\label{lemma:kb_smoother}
    The smoother matrix $\bm S$ can be written as 
    \[\bm S = \bm X \bm A \pt{\bm I_p + \sum_{l \in \mathcal{L}} \lambda_l \text{diag}(\bm s_l)}^{-1} (\bm X \bm A)^\top \bm W\]
for some $\bm A$, $\bm s_0$, $\bm s_{01}$, and $\bm s_{1}$ which do not depend on any $\lambda_l$.
\end{lemma}
The proof of Lemma~\ref{lemma:kb_smoother} is given in the Supplementary Materials, Section \ref{subsec:s_kb_sp}. 
Given Lemma~\ref{lemma:kb_smoother}, the iCV criterion can be approximated as a generalized cross validation (GCV) criterion and further simplified.
Let $\widehat{\bm C}_i$ and $\bm X_i$ be the rows of $\widehat{\bm C}$ and $\bm X$, respectively, that correspond to subject $i$.
Define $\bm F_i = \bm X_i \bm A$, $\bm F = \bm X \bm A$, $\tilde{\bm F} = \bm F^\top \bm W$, $\bm f_i = \bm F_i^\top \widehat{\bm C}_i$, $\tilde{\bm f} = \tilde{\bm F} \widehat{\bm C}$, $\bm J_i = \bm F_i^\top \bm W_i \widehat{\bm C}_i$, $\bm L_i = \bm F_i^\top \bm F_i$, $\tilde{\bm L}_i = \bm F_i^\top \bm W_i \bm F_i$.
Let $\tilde{\bm D} = \br{\bm I + \sum_l \lambda_l \text{diag}(\bm s_l)  }^{-1}$, and let $\tilde{\bm d}$ be its diagonal. 
Let $\odot$ denote the Hadamard product of two matrices of the same dimensions (i.e., elementwise multiplication).
Define $\bm g = \sum_{i = 1}^I \bm J_i \odot \bm f_i$ and $\bm G = \sum_{i = 1}^I (\bm J_i \tilde{\bm f}^\top) \odot \bm L_i$.
Proposition~\ref{prop:gcv} below simplifies iGCV for computational efficiency, where only one element, $\tilde{\bm d}$, depends on $\lambda_0$, $\lambda_{01}$, or $\lambda_1$; the remaining matrices (i.e., $\bm C, \tilde{\bm f}, \bm{f}, \bm F, \bm g, \bm G, \bm L_i, \tilde{\bm L}_i$) can all be precomputed.

\begin{proposition}
\label{prop:gcv}
The iCV criterion can be approximated as iGCV and simplified as follows:
    \begin{align*}
    \text{iGCV} &= \norm{\bm C}^2 - 2 \tilde{\bm d}^\top (\tilde{\bm f} \odot \bm{f}) + (\tilde{\bm f} \odot \tilde{\bm d})^\top (\bm F^\top \bm F) (\tilde{\bm f} \odot \tilde{\bm d}) + 2 \tilde{\bm d}^\top \bm g - 4 \tilde{\bm d}^\top \bm G \tilde{\bm d}  \\
    &\quad + 2 \tilde{\bm d}^\top \br{\sum_{i = 1}^I \set{\bm L_i (\tilde{\bm f} \odot \tilde{\bm d}) } \odot \set{\tilde{\bm L}_i (\tilde{\bm f} \odot \tilde{\bm d})}}.
\end{align*}
\end{proposition}

The proof of Proposition~\ref{prop:gcv} is given in Supplementary Materials, Section~\ref{subsec:s_kb_sp}.
With a grid search in three dimensions, the smoothing parameters $\lambda_0$, $\lambda_{01}$, and $\lambda_1$ are selected to minimize iGCV.

\subsubsection{Eigendecomposition}
\label{subsubsec:eigen_kb}

To obtain positive semi-definite estimates of the covariance functions and facilitate the KKL expansions in Equation~\eqref{eq:kkl_model}, the estimated covariance functions are concatenated into a block matrix 
$\bm{\widehat K}^R_Z = \begin{pmatrix}
    \bm{\widehat{K}}^R_0 & \bm{\widehat{K}}^R_{01} \\
    (\bm{\widehat{K}}^R_{01})^\top & \bm{\widehat{K}}^R_1
\end{pmatrix}$.
An eigendecomposition of $\bm{\widehat K}^R_Z$ yields eigenvalues $\widehat \lambda_n$ and eigenvectors $\bm{\widehat \phi}_n$  for $n = 1, ..., L$.
Following the literature \citep{yao2005fpca, greven2010lfpca, cui2023fmfpca}, negative eigenvalues are trimmed to 0 to produce the positive semi-definite matrix estimate
$\bm{\widehat K}_Z = \sum_{n = 1}^L \max(\lambda_n, 0) \bm{\widehat \phi}_n (\bm{\widehat \phi}_n)^\top$.
Estimates of $K_{0}(s,u)$, $K_{01}(s,u)$, and $ K_1(s,u)$ are then given by the corresponding blocks of $\bm{\widehat K}_Z$.

After trimming, the eigendecomposition is truncated to a finite-dimensional expansion based on proportion of variability explained (PVE) \citep{di2009mfpca}.
We retain the first $N_Z$ principal components, where $N_Z$ is large enough such that $\text{PVE} = \frac{\sum_{n = 1}^{N_Z} \lambda_n^Z}{\sum_{n = 1}^{L} \lambda_n^Z}$ is greater than a prespecified threshold.
Each vector $\widehat{\bm \phi}_n$ contains two blocks of $L$ entries each. Estimates for the eigenfunctions $\phi_n^0(s)$ are given by entries 1 to $L$, and estimates for $\phi_n^1(s)$ are given by entries $L + 1$ to $2L$.

\subsection{Within-Subject Covariance Estimation}
\label{subsec:kw}

We estimate the within-subject covariance function $K_W(s,u)$ using a similar penalized spline regression framework in as Section~\ref{subsec:kb} with two differences: (1) the residual cross-products are constructed from pairs of observations within the same visit and adjusted to target $K_W(s,u)$, and (2) the error variance $\sigma^2$ is also estimated in the regression.

Specifically, for Equation~\eqref{eq:expected_value}, when $j_1 = j_2 = j$, we have 
$
    \E \br{C_{i j j k_1 k_2}} 
    = K_0(s_{i j k_1}, s_{i j k_2}) + T_{ij} K_{01}(s_{i j k_1}, s_{i j k_2}) + T_{ij} K_{01}(s_{i j k_2}, s_{i j k_1}) + T_{ij}^2 K_{1}(s_{i j k_1}, s_{i j k_2}) + K_W(s_{i j k_1}, s_{i j k_2}) + \sigma^2 \delta_{k_1 = k_2}.
$
From Section \ref{subsec:kb}, we have obtained estimates for $K_0(s_{i j k_1}, s_{i j k_2})$, $K_{01}(s_{i j k_1}, s_{i j k_2})$, $K_{01}(s_{i j k_2}, s_{i j k_1})$, and $K_{1}(s_{i j k_1}, s_{i j k_2})$.
Let 
$\widehat A_{i j j k_1 k_2} = \widehat C_{i j j k_1 k_2} - \widehat K_0(s_{i j k_1}, s_{i j k_2}) - T_{ij} \widehat K_{01}(s_{i j k_1}, s_{i j k_2}) - T_{ij} \widehat K_{01}(s_{i j k_2}, s_{i j k_1}) - T_{ij}^2 \widehat K_{1}(s_{i j k_1}, s_{i j k_2}).$
It follows that the set $\set{\widehat A_{i j j k_1 k_2} | k_1 \not = k_2}$
is a collection of estimators of $K_W(\cdot, \cdot)$, and the set
$\set{\widehat A_{i j j k_1 k_2} | k_1 = k_2}$ is a collection of estimators of $K_W(\cdot, \cdot) + \sigma^2$.
We model $K_W(s,u)$ with tensor product splines
\[H_W(s,u) = \sum_{1 \leq \kappa \leq c, 1 \leq \ell \leq c} \theta^W_{\kappa \ell} B_{\kappa}(s) B_{\ell}(u),\] 
where $\bm{\Theta}^W = (\theta_{\kappa \ell}^W)_{1 \leq \kappa \leq c, 1 \leq \ell \leq c}$ is the coefficient matrix, and fit the model using penalized weighted least squares.
The details of constructing the penalized WLS model and choosing the smoothing parameter $\lambda$ are left to the Supplementary Materials, Section~\ref{sec:s_kw}, as they are similar to \cite{sparse_face}. Compared to Section~\ref{subsec:kb}, an additional variable $\sigma^2$ is added to the vector of regression coefficients, and an extra column $\bm \delta$ containing the indicators $\delta_{k_1 k_2}$ is added to the design matrix.
The estimated within-subject covariance function evaluated at $\mathcal{S}^*$ is $\widehat{\bm K}^R_W = \bm B^* \widehat{\bm \Theta}_W (\bm B^*)^\top$. 

The eigenfunctions $\psi_n(s)$ and eigenvalues $\lambda_n^W$ are estimated based on an eigendecomposition of $\widehat{\bm K}^R_W$. Negative eigenvalues are trimmed.
The expansion is truncated to $N_W$ components by PVE, as discussed in Section~\ref{subsubsec:eigen_kb}, to produce a low-dimensional representation.

\subsection{Score Prediction}
\label{subsec:score}

Once the eigenfunctions and eigenvalues are estimated, the LFPCA model in Equation~\eqref{eq:kkl_model} reduces to a mixed effects model, which we solve using Henderson's mixed model equations (MME; \cite{henderson1975mme}).
MME were introduced to FDA by \cite{cui2023fmfpca}; the main adaptation for our sparse data setting is that eigenfunctions are only evaluated at observed time points rather than on a common dense grid.
The details on constructing the matrices to estimate the scores $\xi_{i n_1}$ and $\zeta_{i j n_2}$ are given in the Supplementary Materials, Section~\ref{sec:s_mme}.
After score prediction, we substitute the mean functions and eigenfunctions in Equation \eqref{eq:kkl_model} to predict $Y_{ij}(s)$ at any desired time points in $\mathcal{S}$.

\section{Simulation}
\label{sec:simulation}

We conduct extensive simulations to evaluate the proposed methods in estimating covariance functions and eigenfunctions as well as predicting curves based on sparse observations.
We separate the simulations into three settings based on what methods are available.
Simulation 1 is in the sparse longitudinal setting, where, to our knowledge, only our proposed Sparse LFPCA (SLFPCA) is applicable.
Simulation 2 compares a multilevel version of the proposed approach (SLFPCA-M) to two existing MFPCA methods: MFPCA-SC \citep{di2009mfpca} and Fast MFPCA \citep{cui2023fmfpca}.
Simulation 3 compares the proposed approach to the LFPCA methods by \cite{greven2010lfpca} (LFPCA-G) with dense and complete data, as the available implementation of \cite{greven2010lfpca} does not accommodate missing data.
Due to space considerations, we summarize results from Simulation 2 and 3 and provide full results in the Supplementary Materials, Sections \ref{subsec:s_mfpca_sim} and \ref{subsec:s_greven_sim}.
For all simulations, we set the weight matrices $\bm W = \bm I$. In all plots, outliers are omitted.

\subsection{Simulation 1: Sparse LFPCA}
\label{subsec:sim_lpfca}

\subsubsection{Simulation design}
\label{subsubsec:simulation_design}

Simulation 1 considers the performance of SLFPCA as the number of subjects, visits, points per curve, and error variance change.
Data is generated according to the KKL decomposition in Equation~\eqref{eq:kkl_model}, with
$\mu(s, T_{ij})= 0$ for simplicity and $N_Z = N_W = 2$ eigenfunctions at each level. 
The eigenfunctions are
\begin{align*}
    &\phi^0_1(s) = \sqrt{\frac{2}{3}} \sin(\pi s) \quad 
    &&\phi^1_1(s) =  \sqrt{\frac{2}{3}} \quad  
    &&\psi_1(s) =  1 \\
    &\phi^0_2(s) = \cos(\pi s) \quad 
    &&\phi^1_2(s) = \sin(3 \pi s) \quad  
    &&\psi_2(s) = \sqrt{3} (2s - 1)
\end{align*}

Note that $\set{\phi_1^0(s), \phi_2^0(s)}$, $\set{\phi_1^1(s), \phi_2^1(s)}$, and $\set{\psi_1(s), \psi_2(s)}$ are each orthonormal sets.
In addition, the random intercept eigenfunctions $\phi_1^0(s)$ and $\phi_2^0(s)$ and the random subject eigenfunctions $\phi_1^1(s)$ and $\phi_2^1(s)$ are mutually orthogonal.
The true eigenvalues are $\lambda^{Z}_{n} = \lambda^{W}_{n} = \pt{\frac{1}{2}}^{n-1}$ for $n = 1, 2$.
For each of the $I$ subjects, the number of visits is drawn from a Poisson distribution with mean $J^*$, and for each visit, the number of observations $m_{ij}$ is drawn from a Poisson distribution with mean $m^*$, where $I$, $J^*$, and $m^*$ are simulation parameters specified below.
Visit times $T_{ij}$ are drawn from Unif(0, 1).
To create irregular sampling points for each curve, the $m_{ij}$ locations $s_{ijk}$ are drawn from a Unif(0,1) distribution.
The subject-level scores $\xi_n$ are drawn from $N(0, \lambda_n^Z)$, and the visit-level scores $\zeta_n$ are drawn from $N(0, \lambda_n^W)$.
White noise drawn from $N(0, \sigma^2)$ is added to each observation, where $\sigma^2$ is a simulation parameter.
Although observed sampling times for each curve are irregular and sparse, the eigenfunctions and curves are also generated on a grid $\mathcal{S}^*$ of $L = 1000$ points evenly spaced in $\mathcal{S} = [0,1]$ to compare estimates to.

The base simulation parameters are $I = 350$, $J^* = 30$, $m^* = 5$, and $\sigma^2 = .85$. One parameter is varied at a time; the variations are $I = 100, 800, 1200$, $J^* = 3, 8, 20$, $m^* = 3, 8$, and $\sigma^2 = .1, 2$. 
These parameters are in part motivated by the SensorKit data, in which $I = 365$, $J^* = 27.5$, $m^* = 4.2$.
For each set of simulation parameters, 300 replicate datasets are generated and analyzed. The main metrics of interest are the integrated squared error (ISE) of the covariance functions and eigenfunctions, as well as the mean squared error (MSE) of predictions $Y_{ij}(s)$.
Additional metrics, namely squared error of the eigenvalues and noise variance, are reported in the Supplementary Materials Section~\ref{subsec:s_additional_sim_lfpca}.

\subsubsection{Simulation results}
\label{subsubsec:simulation_results}

\begin{figure}
    \centering
    \includegraphics[width=\linewidth]{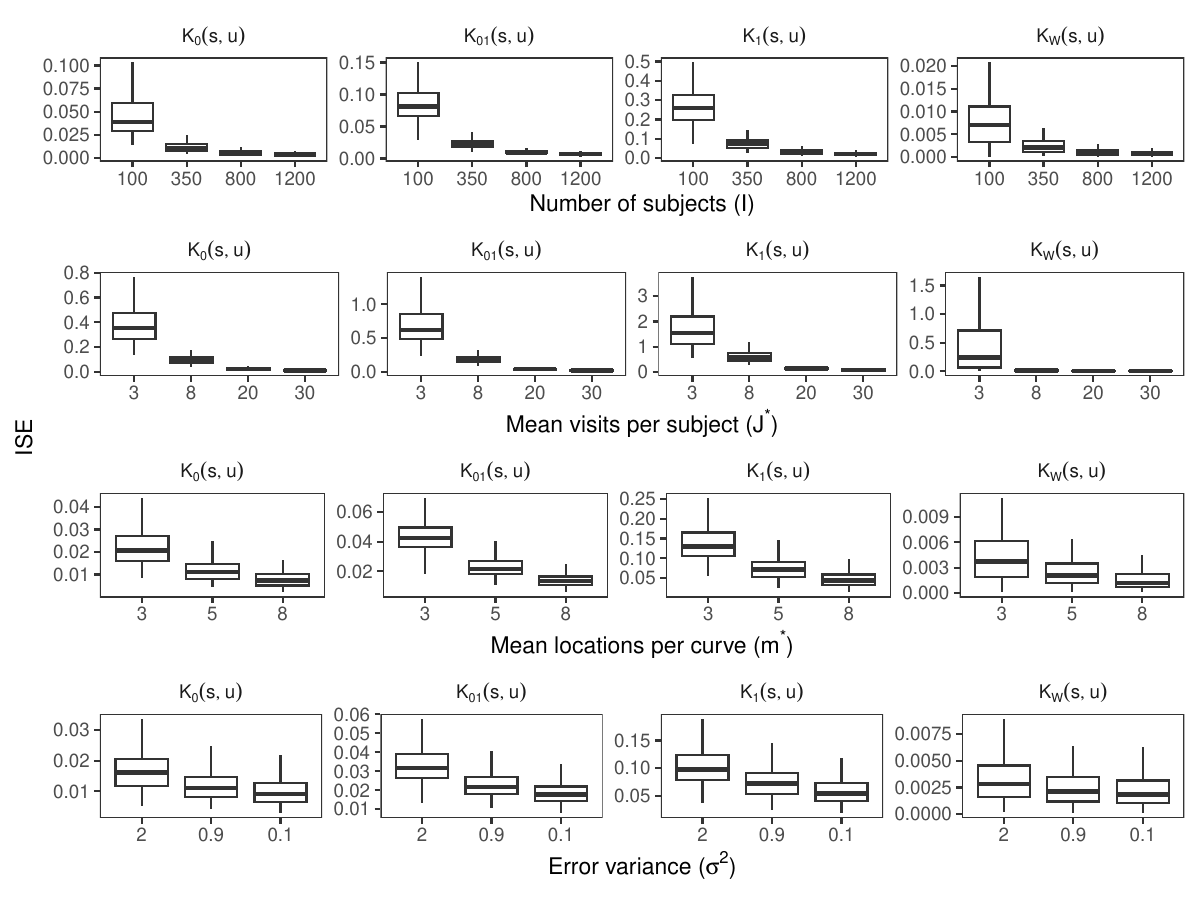}
    \caption{Boxplots of integrated squared error (ISE) for covariance functions in Simulation 1 (sparse LFPCA). Each row varies one parameter (number of subjects $I$, mean visits per subject $J^*$, mean observations per curve $m^*$, or error variance $\sigma^2$) while holding others at baseline values ($I = 350$, $J^* = 30$, $m^* = 5$, and $\sigma^2 = .85$). Columns correspond to $K_0(s,t)$, $K_{01}(s,t)$, $K_1(s,t)$, and $K_W(s,t)$.
    }
    \label{fig:cov_ise_plot}
\end{figure}

Overall, we found that all metrics improve with the number of subjects $I$, the mean number of visits $J^*$, the mean number of locations per curve $m^*$, and lower error variance $\sigma^2$, as expected.
The distribution of ISE of the estimated covariance functions across simulation scenarios is shown in Figure \ref{fig:cov_ise_plot}.
Each row of plots represents a different parameter being varied, and the four columns represent the four covariance functions.
Across all scenarios, the estimation of the within-subject covariance function $K_W(s,u)$ is more accurate than the other three covariance functions---the y-axis scale is an order of magnitude lower. 
This aligns with expectations because compared with the between-subject estimation, (1) there are fewer parameters to estimate, and (2) there are many more levels of random effects (i.e. more subject-visits than subjects).
After $K_W(s,u)$, estimation accuracy is highest for $K_0(s,u)$, followed by $K_{01}(s,u)$, and lowest for $K_1(s,u)$, which may reflect the difficulty in estimating covariances based on visit time covariates.
Additionally, the median ISE for all covariance functions decreases substantially with higher sample sizes, highlighting the need for increased sample size when estimating multiple covariance components simultaneously.

\begin{figure}
    \centering
    \includegraphics[width=\linewidth]{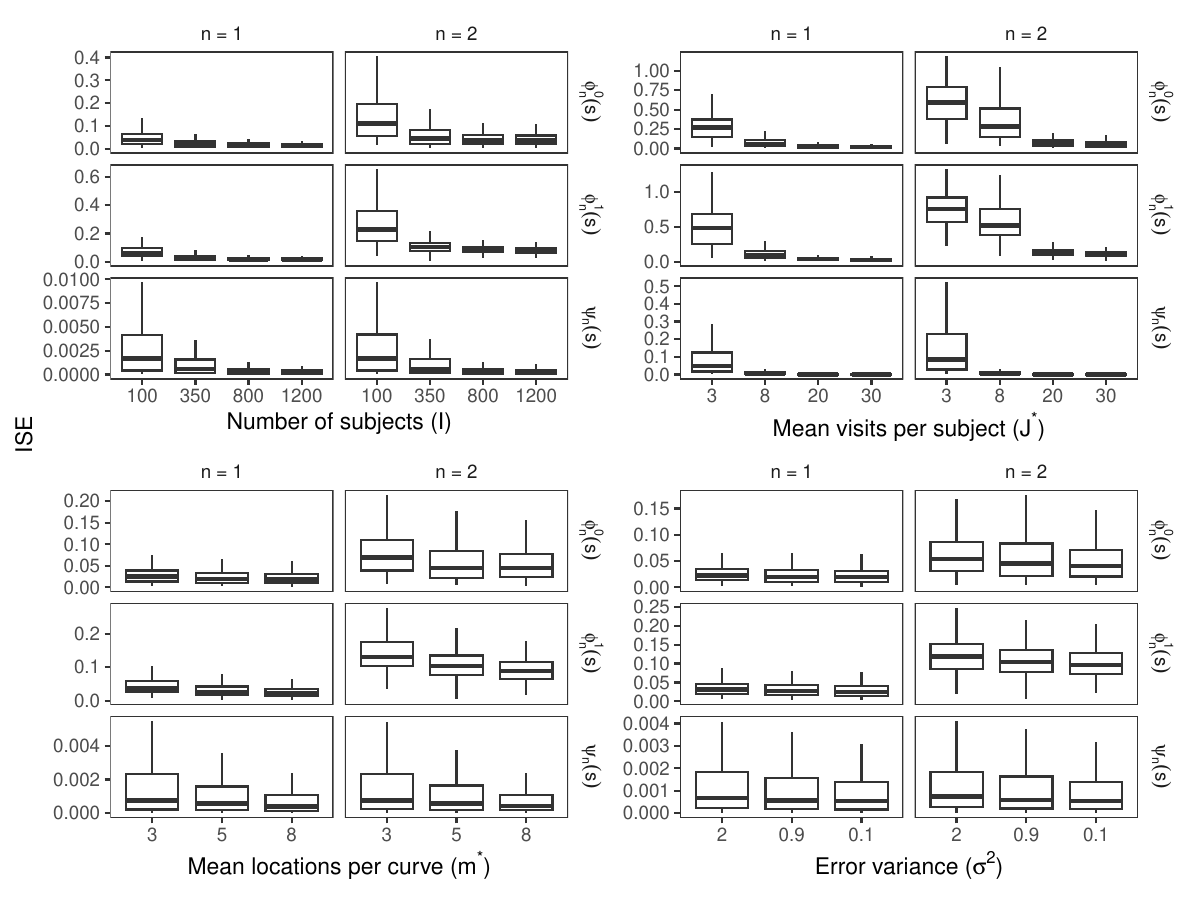}
    \caption{Boxplots of ISEs for eigenfunctions in Simulation 1 (sparse LFPCA). Each quadrant varies one parameter while holding others at baseline. Within each quadrant, the rows of plots represent $\phi_n^0$, $\phi_n^1$, and $\psi_n$; the columns represent the two eigenfunctions for each set.
    The base parameters are $I = 350$, $J^* = 30$, $m^* = 5$, and $\sigma^2 = .85$, and one parameter is varied at a time.
    }
    \label{fig:ef_ise_plot}
\end{figure}

The results in estimating the eigenfunctions are shown in Figure~\ref{fig:ef_ise_plot}. The patterns in covariance function estimation carry over to the eigenfunctions; for example, the eigenfunctions of $K_W(s,u)$ are estimated more accurately than the others.
Additionally, for the subject-level eigenfunctions, the second eigenfunction ($n=2$) is estimated worse than the first, whereas both within-subject eigenfunctions ($\psi_1(s)$, $\psi_2(s)$) are estimated about equally well.

The results for prediction MSE are shown in Supplementary Figure~\ref{fig:s_pred_mse_plot}.
As expected, prediction accuracy improves with more subjects and visits per subject, though with diminishing returns with higher $I$ and $J^*$. More dramatic improvements occur when increasing mean points per curve and reducing error variation. 

In summary, SLFPCA behaves as expected across the range of simulation parameters, capturing the covariance structure at both levels and producing reasonable curve predictions under sparse observations. Establishing reliable performance is crucial in this setting, as many modern longitudinal functional studies involve limited measurements per subject and substantial imbalance across visits.

\subsection{Summary of Simulations 2 and 3}

Simulation 2 evaluates SLFPCA-M as a special case of the proposed framework applied to multilevel data, comparing it against MFPCA-SC \citep{di2009mfpca} and Fast MFPCA \citep{cui2023fmfpca}.
Relative to MFPCA-SC, SLFPCA-M achieved comparable median ISE for the between-subject covariance $K_0(s,u)$ across all simulation parameters, while substantially outperforming MFPCA-SC in estimating the within-subject covariance $K_W(s,u)$, especially when sample sizes ($I$, $J^*$, $m^*$) were small or error variance was large.
Subject-level eigenfunctions were estimated comparably by both methods, but SLFPCA-M consistently outperformed MFPCA-SC on within-subject eigenfunctions. 
Prediction MSE was comparable between the two methods, with slight advantages for SLFPCA-M when $J^*$ was small. 
Compared to Fast MFPCA, SLFPCA-M achieved substantially higher estimation accuracy, particularly for the within-subject covariance and eigenfunctions, at the cost of increased computation time. 
Full results are presented in Supplementary Section~\ref{subsec:s_mfpca_sim}.

Simulation 3 compares SLFPCA to LFPCA-G \citep{greven2010lfpca} in the complete and dense data setting.
For covariance estimation, SLFPCA achieved better performance for $K_0(s,u)$, while LFPCA-G performed better for $K_{01}(s,u)$ and $K_1(s,u)$; estimation of $K_W(s,u)$ was comparable between the two methods. 
SLFPCA had an advantage in estimating $\phi^0_1(s)$ and $\phi^1_1(s)$, while LFPCA-G performed better for $\phi^0_2(s)$. Both methods estimated within-subject eigenfunctions $\psi_n(s)$ similarly. SLFPCA substantially  outperformed LFPCA-G in curve prediction across all simulation scenarios. 
Full results are presented in Supplementary Section~\ref{subsec:s_greven_sim}.

In summary, the proposed approach provides an accurate, scalable, and unified framework for both multilevel and longitudinal sparse FPCA, with comparable or superior performance compared to existing methods.
Nonetheless, the flexibility to estimate cross-covariances in the sparse longitudinal setting is a key advantage not available in existing methods.

\section{Application to Intern Health Study SensorKit Data}
\label{sec:application}

\subsection{Background and Data Description}
\label{subsec:data}

Mental fatigue poses a significant risk to workplace safety and productivity, making it important to detect subtle cognitive declines early enough to intervene \citep{ricci2007fatigue, mccormick2012surgeon}. 
However, reliable measures of mental fatigue require active participation, limiting their real-world utility. 
To address this gap, \cite{fang2026sensorkit_fatigue} studied typing speed performance as a potential ambulatory indicator of mental fatigue within the Intern Health Study (IHS), an annual cohort study of mental health and stress that follows first-year medical residents in the US from the pre-internship period through the intern year (starting July 1st).
This cohort is well-suited for studying mental fatigue as physician interns routinely work extended hours and face high-stakes decisions.
Participants starting residency in 2023 in the United States who used an iPhone were offered enrollment into the Apple SensorKit arm of the study.
Among other measurements, Apple SensorKit \citep{apple_sensorkit, funk2025sensorkit} allows researchers to passively collect typing speed and accuracy without viewing the content typed, enabling scalable and non-intrusive assessment of typing performance in the large IHS cohort.
While \cite{fang2026sensorkit_fatigue} studied both typing speed and rate of deletions, we focus on typing speed because they found relatively minimal changes in deletion rate.

In addition to Apple SensorKit, participants were also provided their choice of wearable device (Fitbit Charge 4, Inspire 2, or Apple Watch) or compensation if they already owned a Fitbit, Apple, or Garmin watch. From the device, sleep and wake times were extracted.
For each typing session, time awake was calculated as the difference between the timestamp of the typing session initiation (rounded to the nearest half hour for privacy) and the wake time of the most recent sleep episode.
Thus, for each participant for each wake period, we have a sample of typing speed observations (one observation per typing session) as a function of time awake, as shown in Figure~\ref{fig:sk_figure}.
The time of a sleep episode (denoted by $T$ in Figure~\ref{fig:sk_figure}) is the difference between the wake time for that episode and 6:00 a.m. on June 29. For example, for participant 232 (middle column), $T = 23$ implies the participant woke up around 6:00 a.m. on July 22, and $ T = 36.8$ implies that the participant woke up at about 1:12 a.m. on August 5; nonstandard wake times are not rare in this cohort.
The sample size for this analysis is 365 participants.
Data from the first two months of the internship (July and August) are included in the analysis, but the number of sleep episodes per participant varies between 1 and 74, with a median of 28 sleep episodes.
The number of typing sessions per day varies between 1 and 36, with a median of 3 typing sessions a day.
Since the sample's average sleep duration is 6.94 hours, only 6.4\% of observations are beyond 17 hours awake. Because there may be a missing sleep episode not captured by the tracker, we exclude observations beyond 17 hours awake to ensure data quality.

Using the proposed SLFPCA framework, our primary goal is to study between-subject differences in patterns of typing speed, such as when typing speed starts to decline, and to characterize these differences by individual-level covariates such as demographics and mental health.
Additionally, we aim to characterize within-person (i.e., day-to-day) differences based on within-person covariates such as the duration of the preceding sleep episode.
Furthermore, it is of interest to dynamically predict an individual's typing performance as their day progresses.
If estimated typing speed decreases below a certain threshold, informing the individual allows them to plan precautionary measures.

\subsection{Results}
\label{subsec:app_results}

\begin{figure}
    \centering
    \includegraphics[width=0.68\linewidth]{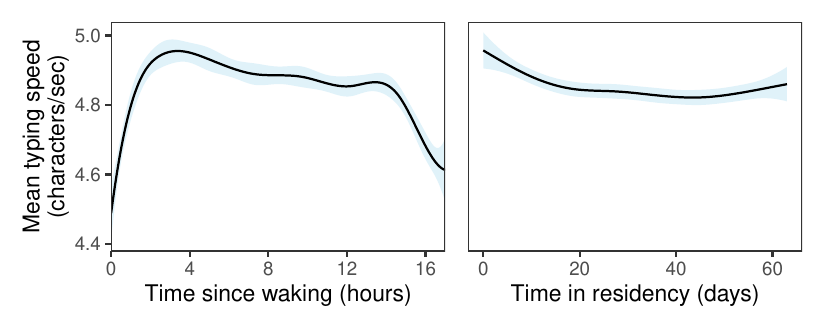}
    \caption{Estimated marginal mean functions in the SensorKit data. The left panel shows $\mu(s)$, the smoothed mean as a function of time awake in hours (averaged across participants and days), and the right panel shows $\mu(T_{ij})$, the mean as a function of days in residency.
    The y-axis is mean typing speed in characters per second.
    The blue ribbons indicate 95\% credible intervals based on the Bayesian posterior covariance of the smoothing coefficients.
    The gray dashed line is the overall average typing speed.}
    \label{fig:mean_sk}
\end{figure}

SLFPCA takes 10.1 minutes to run on the SensorKit typing speed data.
In addition to the population-level mean function $\mu(s, T_{ij})$ used to center the data (step 1), we estimated marginal mean functions $\mu(s)$ and $\mu(T_{ij})$ for visualization (Figure \ref{fig:mean_sk}).
In the left panel, the mean as a function of time of day, $\mu(s)$,  mirrors the average typing speed function in \cite{fang2026sensorkit_fatigue}.
On average, typing speed increases upon waking, peaks at about 3.4 hours after awaking, decreases gradually until about 14 hours, and then drops sharply.
The smaller number of observations beyond 16 hours results in wide confidence intervals after $s = 16$.
The right panel of Figure~\ref{fig:mean_sk} reveals a longitudinal trend not previously examined: on average, typing speed decreased over the course of the study, with the decline most pronounced near the start of the study.
This change over study time motivates analyzing the data using an FPCA model with a longitudinal component.

\begin{figure}
    \centering
    \includegraphics[width=.8\linewidth]{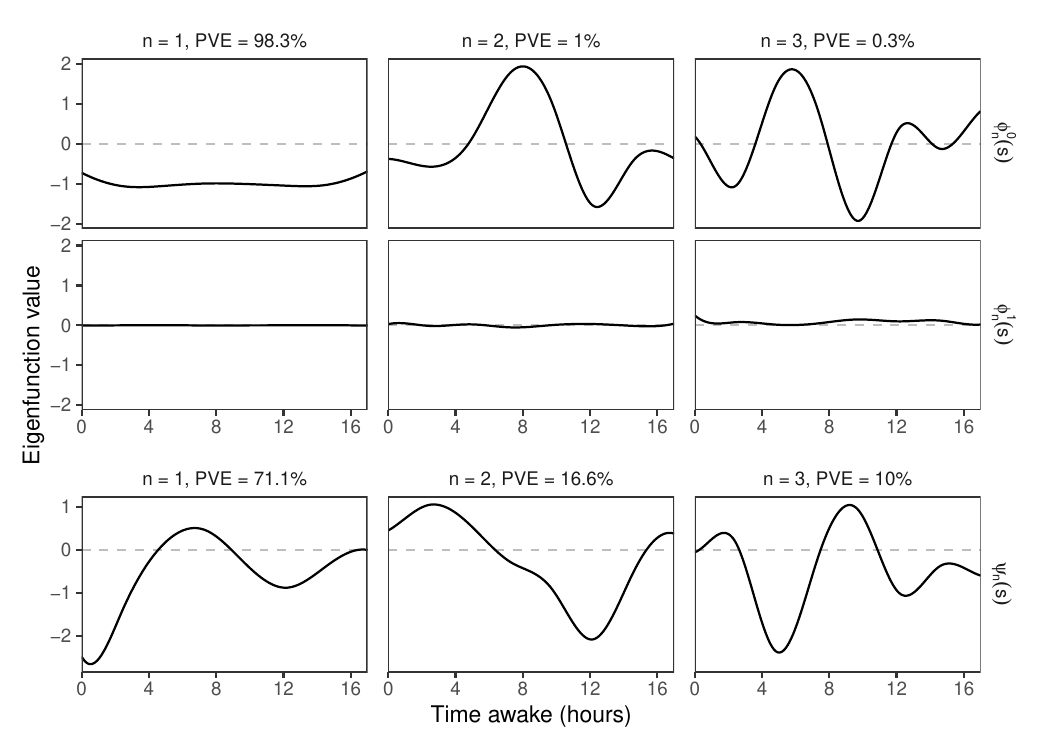}
    \caption{Estimated eigenfunctions from SensorKit typing speed data. The x-axis is time awake in hours. The first row shows the subject-level random intercept eigenfunctions, the second row shows the subject-level random slope, and the third row shows the visit-level eigenfunctions.
    The labels above each panel give the percent of variability explained (PVE) within the level by the eigenfunction.}
    \label{fig:ef_sk}
\end{figure}

Based on the estimated eigenvalues, 90\% of variability is accounted for at the subject-level, and 10\% is at the visit-level, implying larger differences in typing speed trajectories between subjects than across days. 
The top three eigenfunctions for each level and their PVE within the level are shown in Figure~\ref{fig:ef_sk}. 
Between subjects, 98.3\% of variability is explained by the first principal component $\bm \phi_1(s)$.
The first eigenfunction for $K_0(s,u)$, $\phi_1^0(s)$, decreases slightly after waking and then is mostly constant, mirroring the mean function.
For interpretation, Supplementary Figure~\ref{fig:s_ef_pm_plot} shows the population average plus and minus a suitable multiple of each eigenfunction.
Those with positive scores tend to type slower overall, while those with negative scores type faster.
Note that eigenfunctions are identifiable only up to multiplication by -1.
The longitudinal component $\phi_1^1(s)$ of the first PC is approximately zero, indicating that there are minimal subject-specific longitudinal effects, although there are global longitudinal effects.
The other eigenfunctions explain relatively little variability.

Within each subject, the first eigenfunction explains 71.1\% of the variability. It is strongly negative from 0 to 4 hours, representing variation in typing speed initially upon waking.
Visits with negative scores start higher and have a second peak around 12 hours.
The second eigenfunction explains 16.6\% of variability, which is positive for the first 6 hours and negative between 6 and 15 hours. It also captures a dip around 12 hours; visits with positive scores have lower typing speed at that time.
The third eigenfunction represents 10\% of variability; visits with positive scores have lower typing speeds around 4 hours.
The PVE is more spread out across the within-subject eigenfunctions compared to the between-subject eigenfunctions, indicating that the day-to-day patterns are more complex.

Next, we examine univariate associations between subject-level covariates and  scores corresponding to the first subject-level eigenfunction (the only one explaining more than 5\% of variability), and the associations between visit-level covariates and scores corresponding to the first three visit-level eigenfunctions.
We use a t-test for continuous variables and an ANOVA for categorical variables.
We consider the following subject-level covariates: age, sex, ethnicity, marital status, neuroticism, previous history of depression, depression score (PHQ), anxiety score (GAD), stress score (SLE), tobacco use, alcohol use, and cannabis use.
All variables are measured at baseline.
No subject-level variables showed significant associations with scores on the first eigenfunction, i.e., the above variables are not associated with a participant's overall typing speed.

We consider the following visit-level covariates: sleep duration (hours), time of waking (hours), and whether it is a weekend.
Higher scores on the first eigenfunction are negatively associated with sleep duration ($\beta = -0.0004$, $p < .001$) and positively associated with the weekends ($\beta = .0024$, $p = .010$).
That is, shorter sleep durations and weekends tend to be associated with lower typing speeds upon waking.
This is highly interpretable since sleep deprivation is one of the two main factors influencing real-world mental fatigue along with time on task \citep{fang2026sensorkit_fatigue}.
Earlier wake time ($\beta = -0.0002$, $p  = .001$) is associated with scores on the second eigenfunction, i.e., earlier wake times are associated with a lower typing speed around 12 hours.
Waking earlier may shift the dip in alertness to an earlier point, relative to the participant's mean profile.

The proposed method can also produce predicted trajectories based on sparse observations. Figure~\ref{fig:sk_figure} shows a sample of nine trajectories predicted by SLFPCA.
By combining population averages, subject-level effects, and visit-level effects, SLFPCA generates full trajectories over the entire functional domain, even from sparsely observed data. For days with more observations, such as in the top right and middle right panels, the predicted trajectories deviate more from the participant averages, whereas for days with fewer observations, such as all three panels in the third column, the predicted trajectories are shrunk toward the participant's average.

\section{Discussion}
\label{sec:discussion}

In this paper, we proposed sparse LFPCA, the first FPCA framework for sparse and irregular longitudinal functional data. 
Existing LFPCA methods \citep{greven2010lfpca} were developed for densely observed data on a common grid, and existing sparse FPCA methods \citep{sparse_face, di2014sparse_mfpca, cederbaum2016sparse_flmm} do not fully accommodate longitudinal structures.
To address this gap, we introduced a structured penalized spline framework that estimates and smooths multiple covariance components simultaneously, borrowing information across locations.
Simulations show the proposed method successfully estimates covariance functions and eigenfunctions and predicts underlying curves.
Additionally, sparse LFPCA achieves comparable or better performance relative to the LFPCA method designed for dense data by \cite{greven2010lfpca}, and the MFPCA methods by \cite{cui2023fmfpca} and \cite{di2009mfpca, di2014sparse_mfpca}.

Applied to the SensorKit typing speed data from the IHS, sparse LFPCA reveals interpretable subject- and visit-level patterns not captured by the population-level marginal analysis in \cite{fang2026sensorkit_fatigue}, which adopted generalized estimated equations.
Specifically, we found that each participant's baseline typing speed accounts for a majority of between-participant variability. Day-to-day variation only comprises 10\% of total variability but exhibits more complex structures. The primary source of within-subject variation comes from typing speed upon waking, which is associated with daily sleep duration and whether it is a weekend, while the second largest source captures a dip around 12 hours awake and is associated with earlier wake times.
These results demonstrate how sparse LFPCA enables a more granular understanding of typing performance and its relationship with potentially actionable factors.

The proposed method and computational framework can be generalized to further enhance flexibility, speed, and scientific utility. First, while this paper focuses on LFPCA, the proposed framework can be extended to more general functional linear mixed models. Specifically, the regression in Section~\ref{subsubsec:kb_estimation} can flexibly accommodate other random effect structures by modifying the design matrix appropriately.
Second, the methods can be computationally optimized; for example, running an exhaustive grid search to estimate three smoothing parameters can be slow compared to other optimization methods.
Third, developing inferential tools to quantify the uncertainty associated with the estimated eigenfunctions and predicted trajectories would further enhance the practical utility of the proposed framework. Finally, it is of scientific interest to develop a joint inferential framework that correlates the nuanced functional patterns identified at subject- and day-levels with a final cross-sectional outcome or longitudinal outcomes. We leave these directions for future work.

\noindent \textbf{Acknowledgments and funding:} 
We gratefully acknowledge Intern Health Study participants and research coordinators for their contributions, without whom this research would not have been possible.
ZW and SS are partly supported by an NIH grant (R01MH101459) and a Patient-Centered Outcomes Research Institute (PCORI) Project Program Award (ME-2025C1-44006). All statements in this report, including its findings and conclusions, are solely those of the authors and do not necessarily represent the views of the Patient-Centered Outcomes Research Institute (PCORI), its Board of Governors or Methodology Committee.

\bibliographystyle{plainnat}
\bibliography{refs}

\clearpage

\setcounter{section}{0}
\setcounter{subsection}{0}
\setcounter{equation}{0}
\setcounter{figure}{0}
\setcounter{table}{0}

\renewcommand{\thesection}{S\arabic{section}}
\renewcommand{\thesubsection}{S\arabic{section}.\arabic{subsection}}
\renewcommand{\theequation}{S\arabic{section}.\arabic{equation}}

\renewcommand{\thefigure}{S\arabic{figure}}
\renewcommand{\thetable}{S\arabic{table}}

\begin{center}
\huge Supplementary Materials
\end{center}

Section \ref{sec:s_review} expands on limitations of existing methods in the setting of sparse longitudinal functional data.
Section \ref{sec:s_kb} contains details on constructing the regression $\widehat{\bm C} = \bm X \bm \alpha$ and the proofs referenced in the smoothing parameter selection.
Section \ref{sec:s_kw} contains details on constructing the regression for the within-subject covariance estimation and outlines the remainder of the penalized spline estimation framework.
Section \ref{sec:s_mme} describes the mixed models equations for estimating scores.
Section \ref{sec:s_weight} gives the weight matrices for the regressions to increase estimation efficiency.
Section \ref{sec:s_mfpca} describes the modifications to the SLFPCA necessary to conduct MFPCA as a special case.
Section \ref{sec:s_additional_sim} gives additional simulation results for Simulations 1 (sparse LFPCA) and 2 (MFPCA) and presents the design and results for Simulation 3 (LFPCA with complete data).

\section{Limitations of Existing Approaches for Sparse Functional Data}
\label{sec:s_review}

\cite{greven2010lfpca} proposed a method to estimate the LFPCA model in Equation~(2.2) of the main text for dense data, summarized in Algorithm~\ref{alg:greven}.
The main steps are to estimate the fixed effects (step 1), estimate the covariance functions (steps 2 and 3), construct eigendecompositions of the covariance functions (step 4), and predict scores (step 5).
The primary challenge is in accurately estimating the covariance functions.
Specifically, let $r_{ijk} = Y_{ij}(s_{ijk}) - \mu(s_{ijk}, T_{ij})$ denote the residuals after subtracting the fixed effects, and let $C_{i j_1 j_2 k_1 k_2} = r_{ij_1 k_1} r_{i j_2 k_2}$ be cross-products of the residuals, where $j_1, j_2 \in \set{1, ..., J_i}$ index visits and $k_1, k_2 \in \set{1, ..., m_{ij}}$ index observations for each subject.
Based on the model in Equation~(2.1) of the main text, the expected value of $C_{i j_1 j_2 k_1 k_2}$ is given in Equation~(2.3) of the main text, restated here:
\begin{equation}
    \begin{aligned}
    \E \br{C_{i j_1 j_2 k_1 k_2}} 
    &= K_0(s_{i j_1 k_1}, s_{i j_2 k_2}) + T_{ij_2} K_{01}(s_{i j_1 k_1}, s_{i j_2 k_2}) + T_{ij_1} K_{01}(s_{i j_2 k_2}, s_{i j_1 k_1}) \\
    &\quad + T_{ij_1} T_{i j_2} K_{1}(s_{i j_1 k_1}, s_{i j_2 k_2}) + \br{K_W(s_{i j_1 k_1}, s_{i j_2 k_2}) + \sigma^2 \delta_{k_1  k_2} } \delta_{j_1 j_2}.
    \end{aligned}
\end{equation}
where $\delta$ is the Kronecker delta.

\begin{algorithm}
\caption{LFPCA \citep{greven2010lfpca}}\label{alg:greven}
\begin{enumerate}
    \item Estimate the fixed effect surface $\mu(s, T_{ij})$ under the working independence model $Y_{ij}(s) = \mu(s, T_{ij}) + \epsilon_{ij}(s)$.
    \item Using linear regression, estimate the covariance functions $K_{\bm Z}(s,t)$ and $K_W(s,t)$ from residuals $Y_{ij}(s) - \mu(s, T_{ij})$.
    \item Smooth the raw covariance function estimates from step 2; this also provides an estimate for $\sigma^2$.
    \item Construct eigendecompositions of the smoothed covariance functions, and truncate the expansions to provide a low-dimensional representation of $\bm Z_i(s)$ and $W_{ij}(s)$.
    \item Predict subject- and visit- specific scores using Best Linear Unbiased Predictions (BLUPs).
\end{enumerate}
\end{algorithm}

The cross-products $C_{i j_1 k_1 j_2 k_2}$ underpin the covariance estimation in \cite{greven2010lfpca} as well as the proposed methods.
Specifically, based on Equation~(2.3) of the main text, the covariance functions $K_{\bm Z}(s,u)$ and $K_W(s,u)$ can be estimated by fitting a linear regression model at each location $(s, u) \in \mathcal{S}^* \times \mathcal{S}^*$.
For each outcome $C_{i j_1 k_1 j_2 k_2}$, the covariates in the design matrix are $(1, T_{i j_2}, T_{i j_1}, T_{i j_1} T_{i j_2}, \delta_{j_1 j_2})$, associated with the parameters $\{K_0(s, u), K_{01}(s,u), K_{01}(u,s), K_1(s,u), K_W(s,u) + \sigma^2\delta_{k_1k_2}\}$ where $(s, u) = (s_{i j_1 k_1}, s_{i j_2 k_2})$.
This regression yields raw estimates of each covariance function on the grid $\mathcal{S}^*$.
After the regression step, a bivariate smoother is applied separately over each covariance function to obtain the final estimates.

While this approach works well for dense data on a common grid, it faces significant limitations when applied to sparse data.
In the sparse and irregular data setting, the number of observations per curve, $m_{ij}$, is assumed to be relatively small (e.g., about 2 to 5), and the locations $s_{ijk}$ are a random sample in $\mathcal{S}$ rather than being on a grid.
Thus, the number of residual cross-products $C_{i j_1 k_1 j_2 k_2}$ can be very small compared to the number of unknown parameters.
Specifically, if the times $s_{ijk}$ are rounded to the nearest location on a grid of $L$ points, this leads to a total of only $\sum_{i = 1}^I \frac{m_i (m_i + 1)}{2}$ residual cross-products (the number of subjects in the regression model), where $m_i = \sum_{j = 1}^J m_{ij}$ is the total number of observations for subject $i$. 
However, factoring in symmetry constraints, there are $\frac{5 L (L + 1)}{2} + 1$ parameters to estimate in total.
Suppose that $m_i = \rho L$ for all $i$, the model is not estimable if $\frac{I \rho L (\rho L + 1)}{2} < \frac{5 L (L + 1)}{2} + 1$. For example, when $I = 100$ and $L = 100$, the model becomes inestimable when $\rho < .22$, which is a quite common setting for sparse observations. 
Thus, a direct application of \cite{greven2010lfpca} is statistically infeasible for sparse observations.
Moreover, because the regression in \cite{greven2010lfpca} does not borrow information between pairs of locations, the sample size to estimate parameters $K_0(s, u), K_{01}(s, u), K_{01}(u, s), K_1(s, u)$, and $K_W(s, u) + \sigma^2$ for each $(s,u)$ is very limited, leading to high-variance or unstable estimates.

\cite{cederbaum2016sparse_flmm} considered this issue in the sparse functional linear mixed model (FLMM) setting and incorporated smoothing to borrow information across locations.
Instead of a linear regression, they propose fitting a varying coefficient model, where the auto-covariances are estimated as smooth bivariate surfaces.
However, they focus on a simplified crossed design and assume that cross-covariance functions such as $K_{01}(s,u)$ are zero.
In the longitudinal setting, this implies independence between the functional random intercept $Z_{i,0}(s)$ and the functional random slope $Z_{i, 1}(s)$, a strong and unrealistic assumption for most application scenarios.
Moreover, the estimation remains computationally intensive, taking over two hours to fit the crossed model with 707 curves (sample size in \cite{cederbaum2016sparse_flmm}) on a standard computer, as described in their paper. This limits its scalability to larger datasets, such as the Intern Health Study which contains 10,306 daily curves. 

\section{Between-subject covariance estimation}
\label{sec:s_kb}

\subsection{Construction of outcome vector, design matrix, regression coefficients}
\label{subsec:s_kb_matrix}

This section shows the details of constructing the outcome vector $\widehat{\bm C}$, the design matrix $\bm X$, and the vector of regression coefficients $\bm \alpha$ when estimating the subject-level covariance functions.
Let $q_i = \sum_{j_1 = 1}^{J_i - 1} \sum_{j_2 = j_1 + 1}^{J_i} m_{i j_1} m_{i j_2}$ be the total number of cross-visit auxiliary variables for subject $i$.
Define
$\widehat{\bm C}_{ij_1 j_2 k} = 
    \set{\widehat C_{ij_1 j_2 k 1}, \widehat C_{ij_1 j_2 k 2}, ...,\widehat C_{i j_1 j_2 k m_{ij_2}}}^\top \in \R^{m_{i j_2}}$,
$\widehat{\bm C}_{ij_1 j_2} = \set{\widehat{\bm C}_{ij_1 j_2 1}^\top, ..., \widehat{\bm C}_{ij_1 j_2 m_{ij_2}}^\top}^\top \in \R^{m_{ij_1} m_{i j_2}}$, 
$\widehat{\bm C}_{i j_1} = \set{\widehat{\bm C}_{ij_1 (j_1 + 1)}^\top, ...,\widehat{\bm C}_{ij_1 J_i}^\top}^\top \in \R^{m_{j_1}\sum_{j_2 = j_1 + 1}^{J_i} m_{ij_2}}$, and
$\widehat{\bm C}_{i} = \set{\widehat{\bm C}_{i 1}^\top, ..., \widehat{\bm C}_{i J_i}^\top}^\top \in \R^{q_i}$. Then $\widehat{\bm C} = \pt{ \widehat{\bm C}_1^\top, ..., \widehat{\bm C}_I^\top}^\top \in \R^{q}$ is the outcome vector of the regression model, where $q = \sum_{i = 1}^I q_i$.

To construct the design matrix, first define
$\bm{b}(s) = \set{B_1(s), ..., B_c(s)}^\top \in \R^c$.
Considering estimation of $K_0(s,u)$ first, let $\bm \theta_0 = \text{vech} \bm \Theta_0 \in \R^{\frac{c (c + 1)}{2}}$. Let $\bm{G}_c \in \R^{c^2 \times \frac{c(c+1)}{2}}$ be the duplication matrix such that $\text{vec} \bm \Theta_0 = \bm{G}_c \bm{\theta}_0$.
Note that $H_0(s,u) = (\bm b(u) \otimes \bm b(s))^\top \bm{G}_c \bm \theta_0$.
For $j_1 < j_2$, let $\bm{B}_{ij_1 j_2 k_1} = \br{\bm{b}(s_{ij_2 1}), ..., \bm{b}(s_{ij_2 m_{ij_2}})} \otimes \bm{b}(s_{ij_1k_1}) \in \R^{c^2 \times m_{ij_2}}$, 
$\bm{B}_{ij_1 j_2} = \br{\bm B_{i j_1 j_2 1}, ..., \bm B_{i j_1 j_2 m_{i j_1}}}^\top \in \R^{m_{ij_1} m_{i j_2} \times c^2}$,
$\bm{B}_{ij_1} = \br{\bm B_{ij_1(j_1 + 1)}^\top, ..., \bm B_{ij_1 J_i}^\top}^\top \in \R^{m_{i j_1}\sum_{j_2 = j_1 + 1}^{J_i} m_{i j_2} \times c^2}$, 
$\bm B_i = \br{\bm B_{i1}^\top, ..., \bm B_{i J_i}^\top}^\top \in \R^{q_{i} \times c^2}$,
and $\bm B = \br{\bm B_1^\top, ..., \bm B_I^\top}^\top \in \R^{q \times c^2}$. 
The columns of $\bm B_{i j_1 j_2}$ correspond to $\theta^0_{\kappa \ell}$ for $1 \leq \kappa \leq c$ and $1 \leq \ell \leq c$; the same is true for $\bm B_{i j_1}$, $\bm B_i$, and $\bm B$.
Define $\bm X_0 = \bm B \bm G_c \in \R^{q \times \frac{c (c + 1)}{2}}$, the component of the design matrix for estimating $K_0(s,u)$.

The portions of the design matrix for estimating $K_{01}(s,u)$ and $K_1(s,u)$ are constructed similarly.
Define $\bm \theta_{01} = \text{vec} \bm \Theta_{01} \in \R^{c^2}$.
There are $c^2$ coefficients to be estimated because $K_{01}(s,u)$ is not symmetric.
Define $\bm B_{ij_1 j_2}^{s}$ by modifying $\bm B_{i j_1 j_2}$, swapping the column corresponding to $\theta^0_{\kappa, \ell}$ with the column corresponding to $\theta^0_{\ell, \kappa}$.
Define $\bm B^{01}_{i j_1 j_2} = T_{i j_2} \bm B_{i j_1 j_2}$, and construct $\bm B^{01}$ by stacking matrices for the subjects and visit pairs, in the same way as $\bm B$.
Similarly, define $\bm B^{10}_{i  j_1 j_2 } = T_{i j_1} \bm B^s_{i j_1 j_2}$, and construct $\bm B^{10}$ in the same way as $\bm B$.
Then, the portion of the design matrix corresponding to the estimation of $K_{01}(s,u)$ is $\bm X_{01} = \bm B^{01} + \bm B^{10} \in \R^{q \times c^2}$.
Finally, define $\bm \theta_1 = \text{vech} \bm \Theta_1 \in \R^{\frac{c (c + 1)}{2}}$.
Let $\bm B^1_{i j_1 j_2} = T_{i j_1} T_{i j_2} \bm B_{i j_1 j_2}$, and construct $\bm B^1$ in the same way as $\bm B$. Then the component of the design matrix for estimating $K_1(s,u)$ is $\bm X_1 = \bm B^1 \bm G_c \in \R^{q \times \frac{c (c + 1)}{2}}$.

Combining the three portions, the full design matrix is $\bm X = [\bm X_0, \bm X_{01}, \bm X_{1}] \in \R^{q \times p}$, where $p = 2 c^2 + c$, the total number of coefficients to estimate.
The full vector of coefficients to estimate is $\bm \alpha = \pt{\bm \theta_0^\top, \bm \theta_{01}^\top, \bm \theta_1^\top}^\top \in \R^{p}$.

\subsection{Weights for weighted least squares}
\label{sec:s_weight}

Section~2.4.2 of the main text and Section~\ref{subsec:s_kw_weights} reference weight matrices used in weighted least squares: $\bm W_i$ for between-subject covariance estimation and $\widetilde{\bm W}_i$ for within-subject, specified as the inverse of $\Cov (\bm C_i)$ and $\Cov (\bm A_i)$, respectively.
Proposition~2.1 of the main text is used to construct the weight matrices; the proof is below.
Compared to the corresponding proposition in \cite{sparse_face} (Proposition 1), $\bm M_{i j_1 j_2 k_1 k_2}$ contains four additional terms for the between-subject covariances.

\begin{restatedproposition}{2.1}
Define $\bm{M}_{i j_1 j_2 k_1 k_2} = \begin{pmatrix}
    K_0(s_{i j_1 k_1}, s_{i j_2 k_2}) \\
    T_{i j_2} K_{01}(s_{i j_1 k_1}, s_{i j_2 k_2}) \\
    T_{i j_1} K_{01}(s_{i j_2 k_2}, s_{i j_1 k_1})  \\
    T_{i j_1} T_{i j_2} K_1(s_{i j_1 k_1}, s_{i j_2 k_2}) \\
    K_W(s_{i j_1 k_1}, s_{i j_2 k_2}) \delta_{j_1 j_2} \\
    \sigma^2 \delta_{j_1 j_2} \delta_{k_1 k_2}
\end{pmatrix}$. Then
\[\Cov(C_{i j_1 j_2 k_1 k_2}, C_{i j_3 j_4 k_3 k_4}) = \bm{1}^\top \pt{\bm{M}_{i j_1 j_3 k_1 k_3} \otimes \bm{M}_{i j_2 j_4 k_2 k_4} + \bm{M}_{i j_1 j_4 k_1 k_4} \otimes \bm{M}_{i j_2 j_3 k_2 k_3}}.
\]
\end{restatedproposition}

\begin{proof}[Proof of Proposition~2.1 of the main text]

We will use the following lemma. \\
\textbf{Isserlis' Theorem} \cite{isserlis1918formula}: If $(X_1, ..., X_n) \sim MVN(0, \Sigma)$ then
\[\E \br{X_1 \cdots X_m} = \sum_{p \in P_n^2} \prod_{\set{i, j} \in p} \Cov (X_i, X_j),\]
where $P_n^2$ is the set of all possible pairs in $\set{1, ..., n}$.

The covariance between $C_{i j_1 j_2 k_1 k_2}$ and $C_{i j_3 j_4 k_3 k_4}$ is
\begin{align}
    \Cov(C_{i j_1 j_2 k_1 k_2}, C_{i j_3 j_4 k_3 k_4})
    &= \Cov \pt{r_{ij_1 k_1} r_{ij_2 k_2}, 
    r_{ij_3 k_3} r_{ij_4 k_4}} \nonumber \\
    &= \E \br{r_{ij_1 k_1} r_{ij_2 k_2} r_{ij_3 k_3} r_{ij_4 k_4}} 
    - \E \br{r_{ij_1 k_1} r_{ij_2 k_2}} 
     \E \br{r_{ij_3 k_3} r_{ij_4 k_4}}. \label{eq:cov_CC}
\end{align}
The first term in Equation~\ref{eq:cov_CC} above is
\begin{align*}
    \E &\br{r_{ij_1 k_1} r_{ij_2 k_2} r_{ij_3 k_3} r_{ij_4 k_4}} \\
    &= \E \{ 
    \br{Z_{i, 0}(s_{i j_1 k_1}) + T_{ij_1} Z_{i, 1}(s_{i j_1 k_1}) + W_{ij_1}(s_{i j_1 k_1})  + \epsilon_{i j_1}(s_{ij_1 k_1)}} \\
    &\quad \br{Z_{i, 0}(s_{i j_2 k_2}) + T_{ij_2} Z_{i, 1}(s_{i j_2 k_2}) + W_{ij_2}(s_{i j_2 k_2})  + \epsilon_{i j_2}(s_{ij_2 k_2})} \\
    &\quad \br{Z_{i, 0}(s_{i j_3 k_3}) + T_{ij_3} Z_{i, 1}(s_{i j_3 k_3}) + W_{ij_3}(s_{i j_3 k_3})  + \epsilon_{i j_3}(s_{i j_3 k_3})} \\
    &\quad \br{Z_{i, 0}(s_{i j_4 k_4}) + T_{ij_4} Z_{i, 1}(s_{i j_4 k_4}) + W_{ij_4}(s_{i j_4 k_4})  + \epsilon_{i j_4}(s_{i j_4 k_4})}
    \},
\end{align*}
which includes additional terms compared to the proof of Proposition 1 in \cite{sparse_face}.
The next step is to expand the terms and then calculate expectations.
For brevity, the expansion of the terms is omitted; there are 256 terms and, by independence, only 72 nonzero terms.
Expectations that contain 4 terms of either between-subject, within-subject, or error terms are calculated with Isserlis' Theorem; for example,
\begin{align*}
    \E \br{Z_{i,0}(s_{i j_1 k_1}) Z_{i,0}(s_{i j_2 k_2})  
    Z_{i,0}(s_{i j_3 k_3})  Z_{i,0}(s_{i j_4 k_4}) }
    = K_0(s_{i j_1 k_1}, s_{i j_2 k_2}) K_0(s_{i j_3 k_3}, s_{i j_4 k_4})  \\ 
    + K_0(s_{i j_1 k_1}, s_{i j_3 k_3}) K_0(s_{i j_2 k_2}, s_{i j_4 k_4})  \\
    + K_0(s_{i j_1 k_1}, s_{i j_4 k_4}) K_0(s_{i j_2 k_2}, s_{i j_3 k_3}) 
\end{align*}
The remainder of the terms in the expansion can be calculated from the definitions of the covariance functions given in Section~2.4.1 of the main text.
The second term in Equation~\ref{eq:cov_CC} is, based on Equation 2,
\begin{align*}
    \E \br{r_{ij_1 k_1} r_{ij_2 k_2}} 
     \E \br{r_{ij_3 k_3} r_{ij_4 k_4}}
     &= 
    \{
    K_0(s_{i j_1 k_1}, s_{i j_2 k_2}) + T_{ij_2} K_{01}(s_{i j_1 k_1}, s_{i j_2 k_2}) + T_{ij_1} K_{01}(s_{i j_2 k_2}, s_{i j_1 k_1}) \\
    &\quad + T_{ij_1} T_{i j_2} K_{1}(s_{i j_1 k_1}, s_{i j_2 k_2}) + \br{K_W(s_{i j_1 k_1}, s_{i j_2 k_2}) + \sigma^2 \delta_{k_1  k_2} } \delta_{j_1 j_2}
    \} \\
    &\quad  \times \{
    K_0(s_{i j_3 k_3}, s_{i j_4 k_4}) + T_{ij_4} K_{01}(s_{i j_3 k_3}, s_{i j_4 k_4}) + T_{ij_3} K_{01}(s_{i j_4 k_4}, s_{i j_3 k_3}) \\
    &\quad + T_{ij_3} T_{i j_4} K_{1}(s_{i j_3 k_3}, s_{i j_4 k_4}) + \br{K_W(s_{i j_3 k_3}, s_{i j_4 k_4}) + \sigma^2 \delta_{k_3  k_4} } \delta_{j_3 j_4}
    \}.
\end{align*}
Substituting the terms back Equation~\ref{eq:cov_CC} above and subtracting gives Proposition~2.1 of the main text.
\end{proof}

\subsection{Selection of smoothing parameter}
\label{subsec:s_kb_sp}

\begin{restatedlemma}{2.2}
    The smoother matrix $\bm S$ can be written as 
    \[\bm S = \bm X \bm A \pt{\bm I_p + \sum_{l \in \mathcal{L}} \lambda_l \text{diag}(\bm s_l)}^{-1} (\bm X \bm A)^\top \bm W\]
for some $\bm A$, $\bm s_0$, $\bm s_{01}$, and $\bm s_{1}$ which do not depend on any $\lambda_l$.
\end{restatedlemma}

\begin{proof}
Let $\bm M = \bm X^\top \bm W \bm X$ and write $\bm M = \bm M^{1/2} \bm M^{1/2}$ using SVD. Let $\bm R_l = \bm M^{-1/2} \bm Q_l \bm M^{-1/2}$ for each $l$. Since the three $\bm Q$ matrices span different subspaces of $\R^p$, any product of the $\bm Q$ matrices is $\bm 0$,  that is, $\bm Q_0, \bm Q_{01}$, and $\bm Q_{1}$ are mutually commutative.
Because the matrices are symmetric and mutually commutative, they are simultaneously diagonalizable, i.e., there exists one orthogonal $\bm U$ and a vector $\bm s_l$ such that $\bm R_l = \bm U \text{diag}(\bm s_l) \bm U^\top$ for each $l$. Now, let $\bm A = \bm M^{-1/2} \bm U$.
Then 
\begin{align*}
    \bm X^\top \bm W \bm X + \bm Q &=
    \bm M + \sum_{l \in \mathcal{L}} \lambda_l \bm Q_l \\
    &= \bm M^{1/2} \pt{\bm I + \sum_{l \in \mathcal{L}} \lambda_l \bm M^{-1/2} \bm Q_l \bm M^{-1/2}} \bm M^{1/2} \\
    &= \bm M^{1/2} \pt{\bm I + \sum_{l \in \mathcal{L}} \lambda_l \bm R_l} \bm M^{1/2} \\
    &= \bm M^{1/2} \pt{\bm U \bm U^\top + \sum_{l \in \mathcal{L}} \lambda_l \bm U \text{diag}(\bm s_l) \bm U^\top } \bm M^{1/2}  \\
    &= \bm M^{1/2} \bm U \pt{ \bm I + \sum_{l \in \mathcal{L}} \lambda_l \text{diag}(\bm s_l)  } \bm U^\top \bm M^{1/2}.
\end{align*}
Then 
\begin{align*}
    (\bm X^\top \bm W \bm X + \bm Q)^{-1} 
    &= \bm M^{-1/2} \bm U \pt{ \bm I + \sum_{l \in \mathcal{L}} \lambda_l \text{diag}(\bm s_l)  }^{-1} \bm U^\top \bm M^{-1/2} \\
    &= \bm A \pt{ \bm I + \sum_{l \in \mathcal{L}} \lambda_l \text{diag}(\bm s_l)  }^{-1} \bm A^\top,
\end{align*}
where $\bm A = \bm M^{-1/2} \bm U $. Substituting $(\bm X^\top \bm W \bm X + \bm Q)^{-1}$ with $ \bm A \pt{ \bm I + \sum_{l \in \mathcal{L}} \lambda_l \text{diag}(\bm s_l)  }^{-1} \bm A^\top$ gives the target expression. 
\end{proof}

\begin{restatedproposition}{2.3}
The iCV criterion can be approximated as iGCV and simplified as follows:
    \begin{align*}
    \text{iGCV} &= \norm{\bm C}^2 - 2 \tilde{\bm d}^\top (\tilde{\bm f} \odot \bm{f}) + (\tilde{\bm f} \odot \tilde{\bm d})^\top (\bm F^\top \bm F) (\tilde{\bm f} \odot \tilde{\bm d}) + 2 \tilde{\bm d}^\top \bm g - 4 \tilde{\bm d}^\top \bm G \tilde{\bm d}  \\
    &\quad + 2 \tilde{\bm d}^\top \br{\sum_{i = 1}^I \set{\bm L_i (\tilde{\bm f} \odot \tilde{\bm d}) } \odot \set{\tilde{\bm L}_i (\tilde{\bm f} \odot \tilde{\bm d})}}.
\end{align*}
\end{restatedproposition}

\begin{proof}
Let $\bm S_i = \bm X_i (\bm X^\top \bm W \bm X + \bm Q )^{-1} \bm X^\top \bm W$ so that $\bm S_i \widehat{\bm C}_i $ are the fitted values for subject $i$. 
Let $\bm S_{ii} = \bm X_i(\bm X^\top \bm W \bm X + \bm Q )^{-1} \bm X_i^\top \bm W_i$, the diagonal block of $\bm S$ corresponding to auxiliary variables from subject $i$.
We apply Lemma 3.1 from \cite{xu_huang2012_cv} to rewrite the iCV error as
$\text{iCV} = \sum_{i = 1}^I \norm{(\bm I - \bm S_{ii})^{-1} (\bm S_i \widehat{\bm C} -\widehat{\bm C}_i)}^2$.
Similar to \cite{xu_huang2012_cv} and \cite{sparse_face}, we use the approximation $(\bm I - \bm S_{ii})^{-T} (\bm I - \bm S_{ii})^{-1} \approx\bm I + \bm S_{ii} + \bm S_{ii}^\top$ to further simplify iCV, leading to the generalized cross validation (iGCV) criterion
\begin{align*}
    \text{iGCV} &= \sum_{i=1}^I (\bm S_i \widehat{\bm C} - \widehat{\bm C}_i)^\top (\bm I + \bm S_{ii} + \bm S_{ii}^\top) (\bm S_i \widehat{\bm C} - \widehat{\bm C}_i) \\
    &= \sum_{i = 1}^I \br{  (\bm S_i \bm C - \widehat{\bm C}_i)^\top (\bm S_i \widehat{\bm C} - \widehat{\bm C}_i) 
    + (\bm S_i \widehat{\bm C} - \widehat{\bm C}_i)^\top \bm S_{ii} (\bm S_i \widehat{\bm C}-\widehat{\bm C}_i) 
    + (\bm S_i \widehat{\bm C} - \widehat{\bm C}_i)^\top \bm S_{ii}^\top (\bm S_i \widehat{\bm C} - \widehat{\bm C}_i) } \\
    &= \norm{\bm S \widehat{\bm C} - \widehat{\bm C}}^2
    + 2 \sum_{i = 1}^I  (\bm S_i \widehat{\bm C} - \widehat{\bm C}_i)^\top \bm S_{ii} (\bm S_i \widehat{\bm C} - \widehat{\bm C}_i).
\end{align*}

Given Lemma~2.2 of the main text and the modifications for constructing $\bm X$, $\widehat{\bm C}$, and $\bm S$, the remainder of the proof is the same as the proof of Proposition 2 in \cite{sparse_face}, so it is omitted.
\end{proof}

\section{Within-subject covariance estimation}
\label{sec:s_kw}

\subsection{Construction of outcome vector, design matrix, regression coefficients}
\label{subsec:s_kw_matrix}

This section gives the details on constructing the regression model $\widehat{\bm A} = \widetilde{\bm X} \bm \gamma$ to estimate the within-subject covariance $K_W(s, u)$.
The tilde notation below differentiates the matrices constructed for estimation for $K_W(s,u)$ from the matrices constructed in the previous step for the between-subject covariance functions.
Let $\widetilde q_{ij} = m_{ij}(m_{ij} + 1)/2$ be the number of same-visit residual cross-products for subject $i$ at visit $j$ and $\widetilde q_i = \sum_{j = 1}^{J_i} \widetilde q_{ij}$ be the total number of estimates for subject $i$.
To define the outcome vector, let $\widehat{ \bm{A}}_{ijk} = \set{\widehat A_{ijjkk} , \widehat A_{ijjk (k + 1)}, ..., \widehat A_{ijjkm_{ij}}}^\top \in \R^{m_{ij} - k + 1}$, and $\widehat{\bm A}_{ij} = \set{\widehat{\bm A}_{ij1}^\top, ..., \widehat{\bm A}_{ijm_{ij}}^\top}^\top \in \R^{\widetilde q_{ij}}$.
Then let $\widehat{\bm A}_{i} = \set{\widehat{\bm A}_{i1}^\top, ..., \widehat{\bm A}_{i J_i}^\top}^\top \in \R^{\widetilde q_i}$ be a vector of the variables $\widehat{A}_{ijk_1 k_2}$ from the $j$th visit for the $i$th subject, where $k_1 \leq k_2$.
Here, the main adaption from \cite{sparse_face} is to stack visits from a subject into the same vector $\widehat{\bm{A}}_{i}$.
Constructing these vectors at the subject level, instead of the visit level, ensures independence between $\widehat{\bm{A}}_{i}$ terms and makes later calculations of the weight matrix and iGCV statistic simpler.

For the vector of regression coefficients, let $\bm \theta_W = \text{vech} \bm \Theta_W \in \R^{\frac{c (c + 1)}{2}}$; then the parameter vector is $\bm \gamma = (\bm \theta_W^\top, \sigma^2)^\top$.
For the design matrix, let $\widetilde{\bm{B}}_{ijk} = \br{\bm{b}(s_{ijk}), ..., \bm{b}(s_{ijm_{ij}})} \otimes \bm{b}(s_{ijk}) \in \R^{c^2 \times (m_{ij} - k + 1)}$, $\widetilde{\bm{B}}_{ij} = \br{\widetilde{\bm{B}}_{ij1}, ..., \widetilde{\bm{B}}_{ij m_{ij}}}^\top \in \R^{\widetilde q_{ij} \times c^2}$, $\widetilde{\bm{B}}_i = \br{\widetilde{\bm{B}}_{i1}^\top, ..., \widetilde{\bm{B}}_{i J_i}^\top}^\top \in \R^{\widetilde q_i \times c^2}$, and $\widetilde{\bm{B}} = \br{\widetilde{\bm{B}}_1^\top, ..., \widetilde{\bm{B}}_I^\top}^\top \in \R^{\widetilde q \times c^2}$, where $\widetilde q = \sum_{i = 1}^I \widetilde q_i$. 
Let $\bm{\delta}_{ijk} = (1, \bm{0}_{m_{ij} - k}^\top)^\top \in \R^{m_{ij} - k + 1}$, and $\bm{\delta}_{ij} = \set{\bm{\delta}_{ij1}^\top, ..., \bm{\delta}_{ijm_{ij}}^\top}^\top \in \R^{\widetilde q_{ij}}$, and $\bm{\delta}_{i} = \set{\bm{\delta}_{i1}^\top, ..., \bm{\delta}_{i J_i}^\top} \in \R^{\widetilde q_i}$ contain the indicators $\delta_{k_1 = k_2}$ corresponding to  the elements of $\widehat{\bm{A}_i}$.
Combine these matrices across individuals as $\widetilde{\bm X}_i = [ \widetilde{\bm B}_i \bm{G}_c, \bm \delta_i]$, $\widetilde{\bm X} = \br{\widetilde{\bm X}_1^\top, ..., \widetilde{\bm X}_I^\top}^\top$, $\widehat {\bm{A}} = \pt{ \widehat{\bm{A}}_1^\top, ..., \widehat{\bm A}_I^\top}^\top$, and $\bm \delta = (\bm \delta_1^\top, ..., \bm \delta_I^\top)^\top$.

\subsection{Weighted least squares}
\label{subsec:s_kw_weights}

Similar to the between-subject covariance estimation, the model is fit with weighted least squares to improve estimation efficiency.
Note that $\widehat A_{i j j k_1 k_2}$ is an estimate for $C_{i j j k_1 k_2}$, so we specify the weight matrix $\widetilde{\bm W_i}$ as the inverse of $\Cov(\bm C_{i})$.
The weight matrix is constructed using Proposition~2.1 of the main text, and again we define
$\widetilde{\bm W}_i^{-1} = (1 - \beta) \Cov (\bm A_i) + \beta \text{diag}\pt{\text{diag}(\Cov (\bm A_i))}$ to ensure numerical stability.
Let $\widetilde{\bm W} = \text{blockdiag}(\widetilde{\bm W}_1, ..., \widetilde{ \bm W}_n)$.

\subsection{Penalized estimation}
\label{subsec:s_kw_penalty}

As in the between-subject estimation, we add a penalty $\lambda ||\bm \Theta_W \bm{D}||^2_F$, where $\bm{D} \in \R^{c \times (c -2)}$ is a second-order differencing matrix, the norm is the Frobenius norm, and $\lambda$ is the smoothing parameter.
Defining $ \bm P = \bm G_c^\top (\bm I_c \otimes \bm D \bm D^\top) \bm G_c$ and $\widetilde{\bm Q} = \begin{pmatrix}
    \bm P & \bm 0 \\
    \bm 0 & 0 
\end{pmatrix}$,
the objective function is
\[\hat{\bm \gamma} = \argmin_{\bm \gamma} \pt{(\widehat{\bm{A}} - \widetilde{\bm{X}} \bm \gamma)^\top \widetilde{\bm W} (\widehat{\bm{A}} - \widetilde{\bm{X}} \bm \gamma) + \lambda \bm \gamma^\top \widetilde{\bm Q} \bm \gamma},\]
and an explicit form for $\widehat{\bm{\gamma}}$ is
\[\widehat{\bm{\gamma}} = \begin{pmatrix}
    \widehat{\bm{\theta}}_W \\ \widehat \sigma^2
\end{pmatrix} =  \pt{\widetilde{\bm{X}}^\top \widetilde{\bm W} \widetilde{ \bm X} + \lambda \widetilde{ \bm Q}}^{-1} \pt{ \widetilde{\bm X}^\top \widetilde{\bm W} \widehat{\bm{ A}}}.\]

The smoothing parameter $\lambda$ is selected by leave-one-subject-out cross validation, as in \cite{sparse_face} and Section~2.4.4 of the main text.
We use the efficient approximation for iCV derived in Proposition 2 of \cite{sparse_face}, modifying the matrices to account for the multi-level structure in the same way that we modified constructing $\widehat{\bm{A}_i}$.

\section{Score prediction with MME}
\label{sec:s_mme}

To estimate the scores $\xi_{i n_1}$ and $\zeta_{ij n_2}$, we first construct the matrix form of the mixed effects model in Equation~(2.2) of the main text, adapting the approach of \cite{cui2022fui} to sparse and longitudinal data.
Let $r_{ijk} = Y_{ij}(s_{ijk}) - \mu(s_{ijk}, T_{ij})$ be the residuals as defined in Section~2.1 of the main text. Let $\bm{r}_{ij} = \set{ r_{ij1}, ..., r_{ijm_{ij}}}^\top \in \R^{m_{ij}}$, and $\bm{r}_{i} = \set{\bm{r}_{i1}^\top, ..., \bm{r}_{i J_i}^\top}^\top \in \R^{m_i}$, where $m_i = \sum_{j = 1}^{J_i} m_{ij}$ is the total number of observations for subject $i$.
Let $\bm{s}_{ij} = \set{s_{ij1}, ..., s_{i j m_{ij}}}^\top \in \R^{m_{ij}}$ and $\bm{s}_i = \set{\bm{s}_{i1}^\top, ..., \bm{s}_{i J_i}^\top}^\top \in \R^{m_i}$. 

The design matrices are defined as follows.
Let $\bm{\phi}^0_{i n_1}$ be the vector of $\phi^0_{n_1}(s)$ evaluated at locations in $\bm{s}_i$ and $\bm{\Phi}^0_i = \pt{\bm{\phi}^0_{i 1} \cdots  \bm{\phi}^0_{i N_Z}} \in \R^{m_i \times N_Z}$; define $\bm{\Phi}^1_i$ analogously.
Let $\bm T_{ij} = T_{ij} \bm J_{m_{ij} \times N_Z}$ and $\bm T_i = (\bm T_{i1}^\top, ..., \bm T_{ij}^\top)^\top \in \R^{m_i \times N_Z} $ be a matrix of visit times.
Let $\bm M_i = \bm{\Phi}^0_i + \bm T_i \odot \bm \Phi^1_i$.
Similarly, let $\bm{\psi}_{i j n_2}$ be the vector of $\psi_{n_2}(s)$ evaluated at times $\bm{s}_{ij}$, let $\bm{\Psi}_{ij} = \pt{\bm{\psi}_{i j 1} \cdots \bm{\psi}_{i j N_W}} \in \R^{m_{ij} \times N_W}$, and let $\bm{\Psi}_i = \text{blockdiag}(\bm{\psi}_{i1}, ..., \bm{\psi}_{i J_i}) \in \R^{m_i \times J_i N_W}$.

Now, let $\bm{\xi}_i = \set{\xi_{i 1}, ..., \xi_{i N_Z}}^\top \in \R^{N_Z}$ be the level 1 scores for the $i$th subject. Similarly, let $\bm{\zeta}_{ij} = \set{\zeta_{i j 1}, ..., \zeta_{i j N_W}}^\top \in \R^{N_W}$; then $\bm{\zeta}_i = \set{\bm{\zeta}_{i1}^\top, ..., \bm{\zeta}_{i J_i}^\top}^\top \in \R^{N_W J_i}$ is the vector of level 2 scores for subject $i$.
Let $\bm{\Lambda}_Z = \text{diag}(\lambda_1^{Z}, ..., \lambda_{N_Z}^{Z}) \in \R^{N_Z \times N_Z}$ and $\bm{\Lambda}_W = \text{diag}(\lambda_1^{W}, ..., \lambda_{N_W}^{W}) \in \R^{N_W \times N_W}$ be the covariance matrices for $\bm{\xi}_i$ and $\bm{\zeta}_{ij}$, respectively; the covariance of $\bm{\zeta}_{i}$ is $\bm{I}_{J_i} \otimes \bm{\Lambda}_W$.
Finally let $\bm{\epsilon}_i$ be the vector of error terms for subject $i$ with covariance $\sigma^2 \bm I_{m_i}$.
In matrix mixed model form, for the $i$th subject, Equation~(2.2) of the main text is
\[\bm{r}_i = \bm{M}_i \bm{\xi}_i + \bm{\Psi}_{i} \bm{\zeta}_i + \bm{\epsilon}_i.\]
The mixed model equations (MME) give the solution
\[\begin{pmatrix}
    \hat{\bm{\xi}}_i \\ \hat{\bm{\zeta}}_i
\end{pmatrix} = 
\begin{pmatrix}
    \bm{M}_i^\top \bm{M}_i + \sigma^2 \bm \Lambda_Z^{-1}  & \bm M_i^\top \bm \Psi_i \\
    (\bm \Psi_i)^\top \bm M_i & (\bm \Psi_i)^\top \bm \Psi_i +  \sigma^2 \bm I_{J_i} \otimes \bm{\Lambda}_W^{-1}
\end{pmatrix}
\begin{pmatrix}
    \bm M_i^\top \bm r_{i} \\
    \bm \Psi_i^\top \bm r_i
\end{pmatrix}.\]
One could also use the best linear unbiased predictor (BLUP) and obtain identical estimates. In the setting of dense data, there are significant computational advantages to using the MME instead of the BLUP \citep{cui2023fmfpca}. However, because the number of points per curve $m_{ij}$ is small in sparse data, both methods perform similarly, so choosing between the MME and BLUP methods is a matter of preference.

\section{MFPCA}
\label{sec:s_mfpca}

The proposed framework can be used to estimate MFPCA as a special case of LFPCA with zero longitudinal component.

\subsection{Model and estimation}
\label{subsec:s_mfpca_model}

For subject $i$, $i = 1, ..., I$ at visit $j$, $j = 1, ..., J_i$, we observe $Y_{ij}(s)$ at locations $\set{s_{ijk}}_{k = 1, ..., m_{ij}} \subset \mathcal{S}$, where $m_{ij}$ is the number of observations for subject $i$ at visit $j$. The MFPCA model is
\[Y_{ij}(s) = \mu(s) + Z_{i}(s) + W_{ij}(s) + \epsilon_{ij}(s),\]
where $\mu(s)$ is the fixed effect surface, $Z_i(s)$ is the subject-level random intercept, $W_{ij}(s)$ is the visit-level random intercept, and $\epsilon_{ij}(s)$ is white noise with zero mean and variance $\sigma^2$.
This is identical to the LFPCA model except that for the subject-level random effects, there is only a random intercept.
Let $K_0(s,u) = \Cov \pt{Z_{i}(s), Z_i(u)}$ be the between-subject covariance function and $K_W(s,u) = \Cov \pt{W_{ij}(s), W_{ij}(u)}$ be the within-subject covariance function; both are symmetric.
The mean function $\mu(s)$ is estimated as a smooth function of $s$ using P-splines, although other fixed effect structures can be used depending on the application.

The process for estimating the between-subject covariance function remains the same, except we only estimate $K_0(s, u)$.
The outcome vector $\bm C$ is constructed exactly as described in Supplementary Section~\ref{subsec:s_kb_matrix}, for LFPCA.
The design matrix $\bm X$ contains only the first block, $X_0$, and correspondingly, the regression coefficients are $\bm \alpha = \bm \theta_0$.
Because only one covariance function is estimated at the between-subject level, the penalty and smoothing parameter selection are the same as in the within-subject covariance estimation described in Supplementary Section~\ref{subsec:s_kw_matrix}.

The within-subject covariance estimation is identical except that the estimators are $\widehat A_{i j k_1 k_2} = \widehat C_{i j k_1 k_2} - \widehat K_0(s_{ij k_1}, s_{i j k_2})$, instead of $\widehat A_{i j k_1 k_2} = \widehat C_{i j k_1 k_2} - \widehat K_0(s_{ij k_1}, s_{i j k_2}) - T_{ij} \widehat K_{01}(s_{i j k_1}, s_{i j k_2}) - T_{ij} \widehat K_{01}(s_{i j k_2}, s_{i j k_1}) - T_{ij}^2 \widehat K_1(s_{i j k_1}, s_{i j k_2})$.

The steps for the eigendecomposition, truncation, and score prediction remain similar.
The KKL expansion of the MFPCA model is
\begin{equation}
    Y_{ij}(s) = \mu(s) 
    + \sum_{n_1 = 1}^{N_Z} \xi_{in_1} \phi^0_{n_1}(s)
    + \sum_{n_2 = 1}^{N_W} \zeta_{i j n_2} \psi_{n_2}(s) + \epsilon_{ij}(s)
\end{equation}
For the score prediction, the only modification to Section~\ref{sec:s_mme} necessary is that $\bm M_i = \bm \Phi_i^0$.

\section{Additional simulation results}
\label{sec:s_additional_sim}

\subsection{Simulation 1: Sparse LFPCA}
\label{subsec:s_additional_sim_lfpca}

Additional simulation results are shown below.

\begin{figure}[H]
    \centering
    \includegraphics[width=0.625\linewidth]{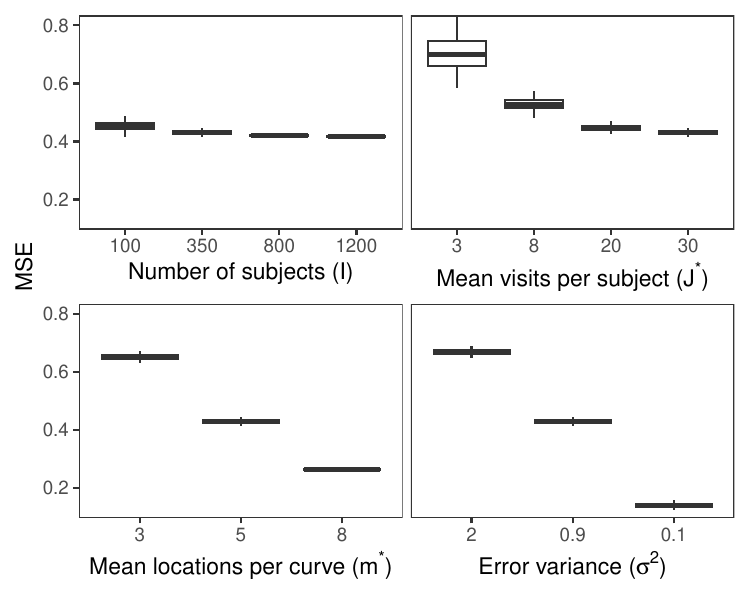}
    \caption{Boxplots of MSEs for predictions in Simulation 1 (sparse LFPCA). Each plot varies one parameter while holding the others at baseline: $I = 350$, $J^* = 30$, $m^* = 5$, and $\sigma^2 = .85$.
    }
    \label{fig:s_pred_mse_plot}
\end{figure}

\begin{figure}[H]
    \centering
    \includegraphics[width=\linewidth]{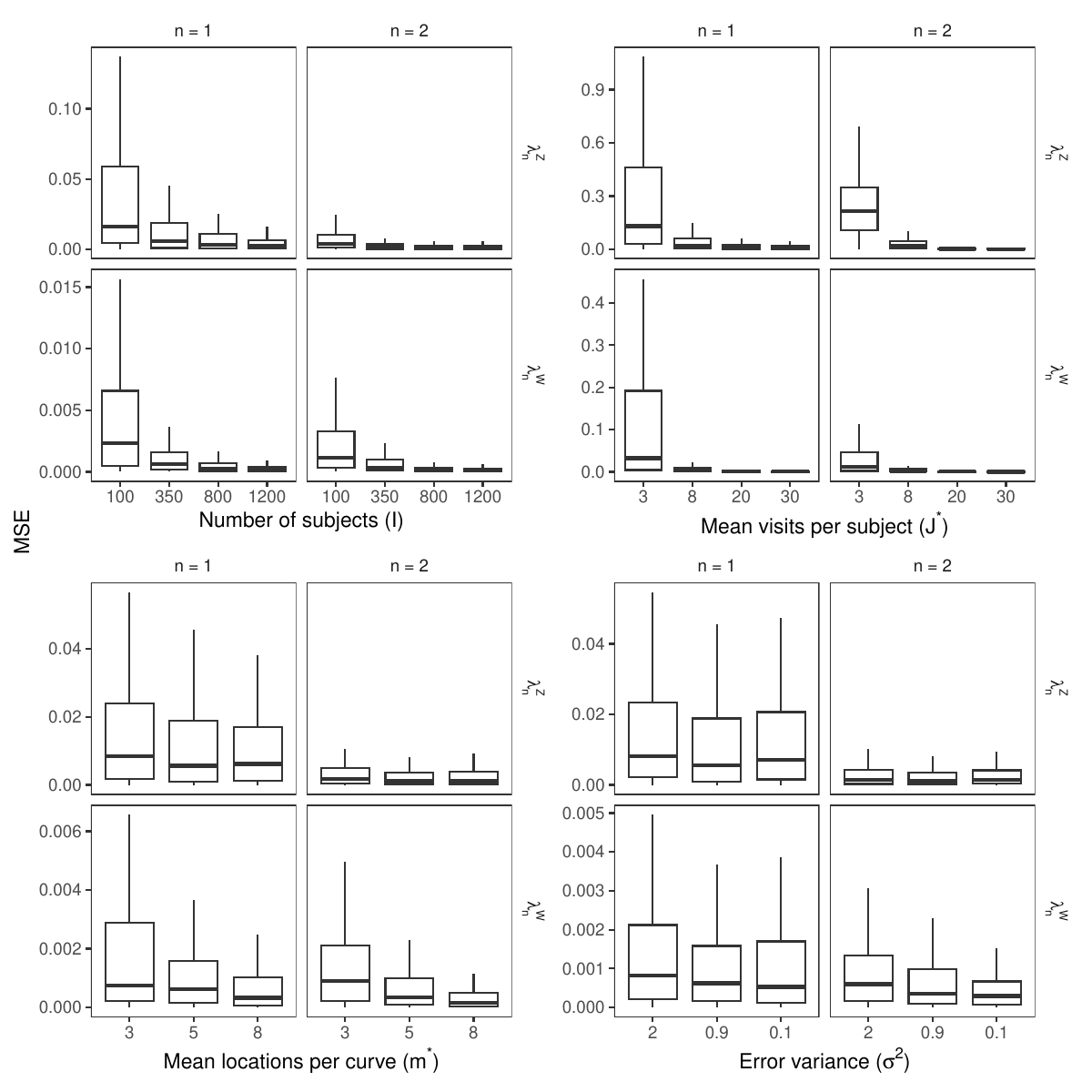}
    \caption{Boxplots of MSEs of eigenvalues in Simulation 1 (sparse LFPCA). 
    Each quadrant varies one parameter(number of subjects $I$, mean visits per subject $J^*$, mean observations per curve $m^*$, or error variance $\sigma^2$) while holding others at baseline values ($I = 350$, $J^* = 30$, $m^* = 5$, and $\sigma^2 = .85$).
    The patterns match those of the eigenfunctions in Figure~2 of the main text. }
    \label{fig:s_ev_mse_plot}
\end{figure}

\begin{figure}[H]
    \centering
    \includegraphics[width=0.625\linewidth]{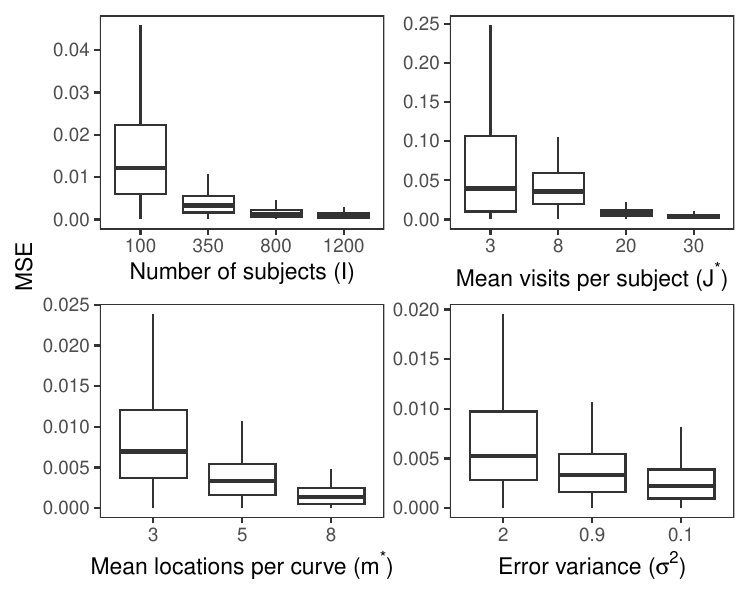}
    \caption{Boxplots of MSE of error variance estimation in Simulation 1 (sparse LFPCA). 
    The base parameters are $I = 350$, $J^* = 30$, $m^* = 5$, and $\sigma^2 = .85$, and one parameter is varied at a time.
    As expected, median MSE decreases when increasing $I$, $J^*$, and $m^*$, or reducing $\sigma^2$.}
    \label{fig:s_mse_sigma2_plot}
\end{figure}

\subsection{Simulation 2: MFPCA}
\label{subsec:s_mfpca_sim}

Since MFPCA can be viewed as a special case of our LFPCA framework, we conduct a simulation study to compare the proposed methods with existing approaches for MFPCA.

\subsubsection{Simulation design}
\label{subsubsec:s_sim_design_mfpca}

Following the design of Simulation 1, data is generated according to the KKL decomposition
\begin{equation}
    Y_{ij}(s) = \mu(s) 
    + \sum_{n_1 = 1}^{N_Z} \xi_{in_1} \phi^0_{n_1}(s)
    + \sum_{n_2 = 1}^{N_W} \zeta_{i j n_2} \psi_{n_2}(s) + \epsilon_{ij}(s)
    \label{eq:mfpca_kkl}.
\end{equation}
The number of visits for each subject $J_i$, number of observations for each visit $m_{ij}$, and scores $\xi_{in}$ and $\zeta_{in}$ are drawn in the same way; no visit times $T_{ij}$ are drawn.
The eigenvalues and eigenfunctions are identical, except the longitudinal components are removed by setting $\phi^1_1(s) = \phi^1_2(s) = 0$, and the subject-level eigenfunctions $\phi^0_1(s)$  and $\phi^0_2(s)$ are accordingly rescaled to have norm 1.
The simulation parameters and performance metrics are the same as in the first study, and 300 replicate datasets are analyzed for each parameter setting.

We compare the proposed method (SLFPCA-M) to two existing approaches: MFPCA-SC via \texttt{mfpca.sc()} \citep{di2009mfpca, di2014sparse_mfpca} and Fast MFPCA via \texttt{mfpca.face()} \citep{cui2023fmfpca}, both available in the \texttt{R} package \texttt{refund}.
MFPCA-SC (``smooth covariance") constructs method of moments estimators of the covariance functions, similar to LFPCA by \cite{greven2010lfpca}, except that the between-subject estimation involves only one covariance function.
Fast MFPCA leverages Fast Covariance Estimation \citep{xiao2016denseface} for each step to achieve computation times orders of magnitude faster than the original MFPCA.
Both MFPCA-SC and Fast MFPCA require data input on a common grid, that is, in wide matrix form. Accordingly, we convert the sparse long form data to a grid $\mathcal{S}^*$ of $L = 1000$ common locations by rounding each location $s_{ijk}$ to the nearest grid point. If multiple locations from the same curve are mapped to the same grid point, the corresponding outcomes are averaged.
Additionally, because Fast MFPCA requires at least $m_{ij} = 4$ observations per curve, visits with fewer observations are excluded when applying this method. As a result, prediction metrics for Fast MFPCA are computed only with the remaining visits.

\subsubsection{Simulation results}
\label{subsubsec:s_sim_results_mfpca}

Across most simulation settings, Fast MFPCA exhibits larger errors. This is not surprising given our sparse simulation design and that Fast MFPCA requires dense observations to improve low-rank approximation performance. To avoid obscuring results for other methods, we present results with Fast MFPCA separately, below.

\begin{figure}[H]
    \centering
    \includegraphics[width=\linewidth]{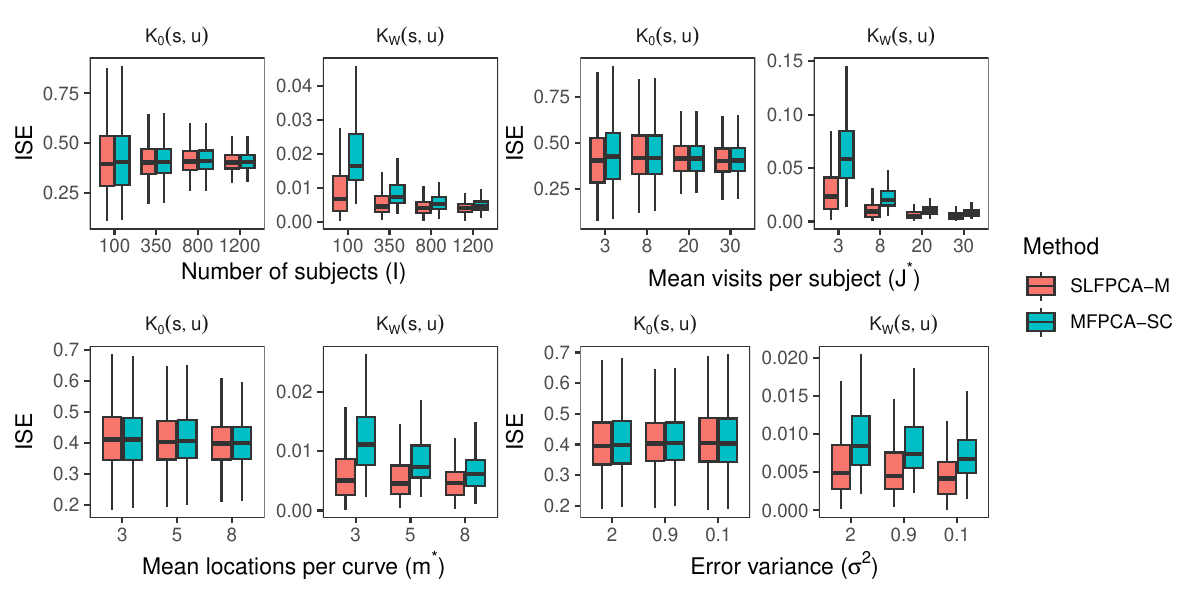}
    \caption{Boxplots of ISE of covariance functions in Simulation 2 (MFPCA). Each quadrant varies one parameter while the other parameters are fixed at baseline ($I = 350$, $J^* = 30$, $m^* = 5$, and $\sigma^2 = .85$). The panels within the quadrant correspond to $K_0(s,t)$ and $K_W(s,t)$.}
    \label{fig:s_cov_ise_mfpca}
\end{figure}

Figure \ref{fig:s_cov_ise_mfpca} compares SLFPCA-M with MFPCA-SC in terms of the ISE of $K_0(s,u)$ and $K_W(s,u)$. 
When estimating $K_0(s,u)$ (left plot in each quadrant), SLFPCA-M achieves median ISE comparable with MFPCA-SC in all scenarios.
Additionally, SLFPCA-M substantially outperforms MFPCA-SC when estimating $K_W(s,u)$ (right plot in each quadrant) across scenarios, especially when $I$, $J$, or $m^*$ are small or $\sigma^2$ is large.
Supplementary Figure \ref{fig:s_cov_ise_mfpca_f} includes Fast MFPCA, which performs slightly worse than the other methods in estimating $K_0(s,u)$ and substantially worse in estimating $K_W(s,u)$, especially when $m^*$ is small or $\sigma^2$ is high.

\begin{figure}[H]
    \centering
    \includegraphics[width=\linewidth]{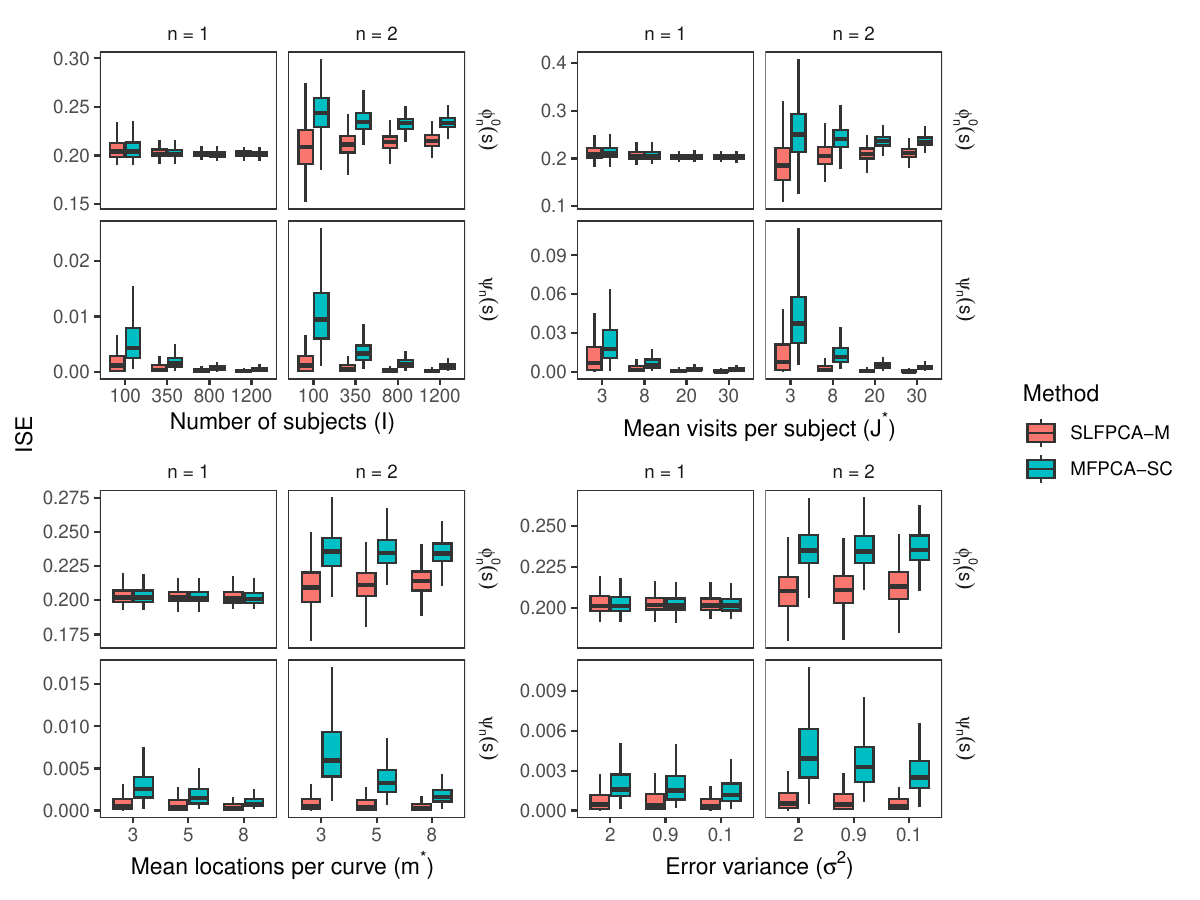}
    \caption{Boxplots of ISE of eigenfunctions in Simulation 2 (MFPCA). Each quadrant varies one parameter while the others are fixed at baseline ($I = 350$, $J^* = 30$, $m^* = 5$, and $\sigma^2 = .85$). Rows correspond to the between-subject and within-subject levels, and columns correspond to the two eigenfunctions at each level.
    }
    \label{fig:s_ef_ise_mfpca}
\end{figure}

Figure \ref{fig:s_ef_ise_mfpca} shows the ISE of estimated eigenfunctions across simulations. Based on the top rows of the four panels, the SLFPCA-M and MFPCA-SC estimate the subject-level eigenfunctions $\phi_1^0(s)$ and $\phi_2^0(s)$ similarly well.
For both within-subject eigenfunctions, the SLFPCA-M consistently outperforms MFPCA-SC across scenarios.
Moreover, for SLFPCA-M, the errors in estimating the second eigenfunction are not substantially higher than the first.
In contrast, MFPCA-SC estimates $\psi_2(s)$ considerably worse than $\psi_1(s)$, especially when $I$, $J^*$, and $m^*$ are small.
Supplementary Figure \ref{fig:s_ef_ise_mfpca_f} shows that Fast MFPCA performs worse than the other two methods, especially in estimating the second eigenfunction at each level.

Figure \ref{fig:s_pred_mse_mfpca} shows the MSE of predictions. SLFPCA-M performs slightly better when $J^*$ is low, but otherwise the difference between the methods is small compared to increasing the mean points per curve or error variance. 

\begin{figure}[H]
    \centering
   \includegraphics[width=0.75\linewidth]{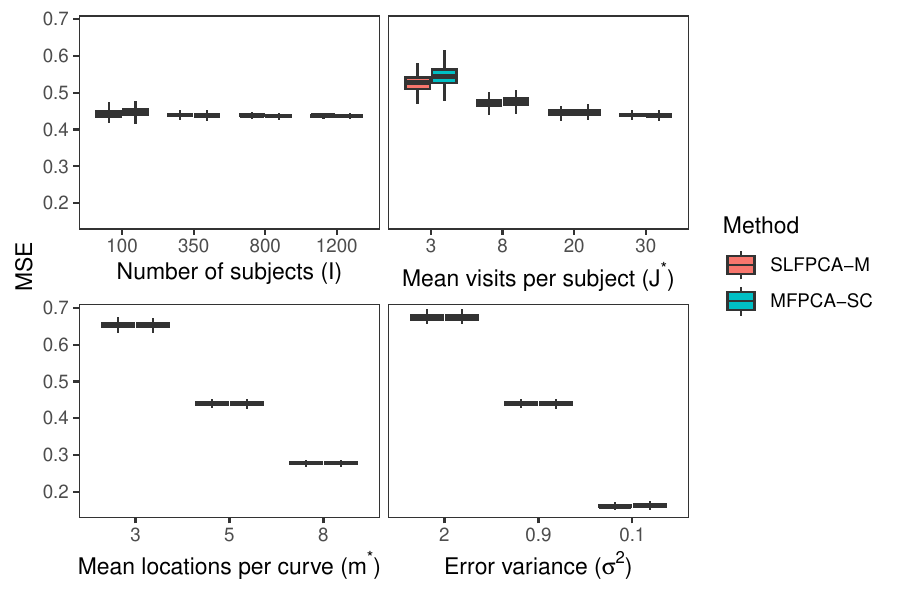}
    \caption{Boxplots of MSEs for predictions in Simulation 2 (MFPCA). Each plot varies one parameter while holding the others at baseline: $I = 350$, $J^* = 30$, $m^* = 5$, and $\sigma^2 = .85$.}
    \label{fig:s_pred_mse_mfpca}
\end{figure}

\begin{figure}[H]
    \centering
    \includegraphics[width=\linewidth]{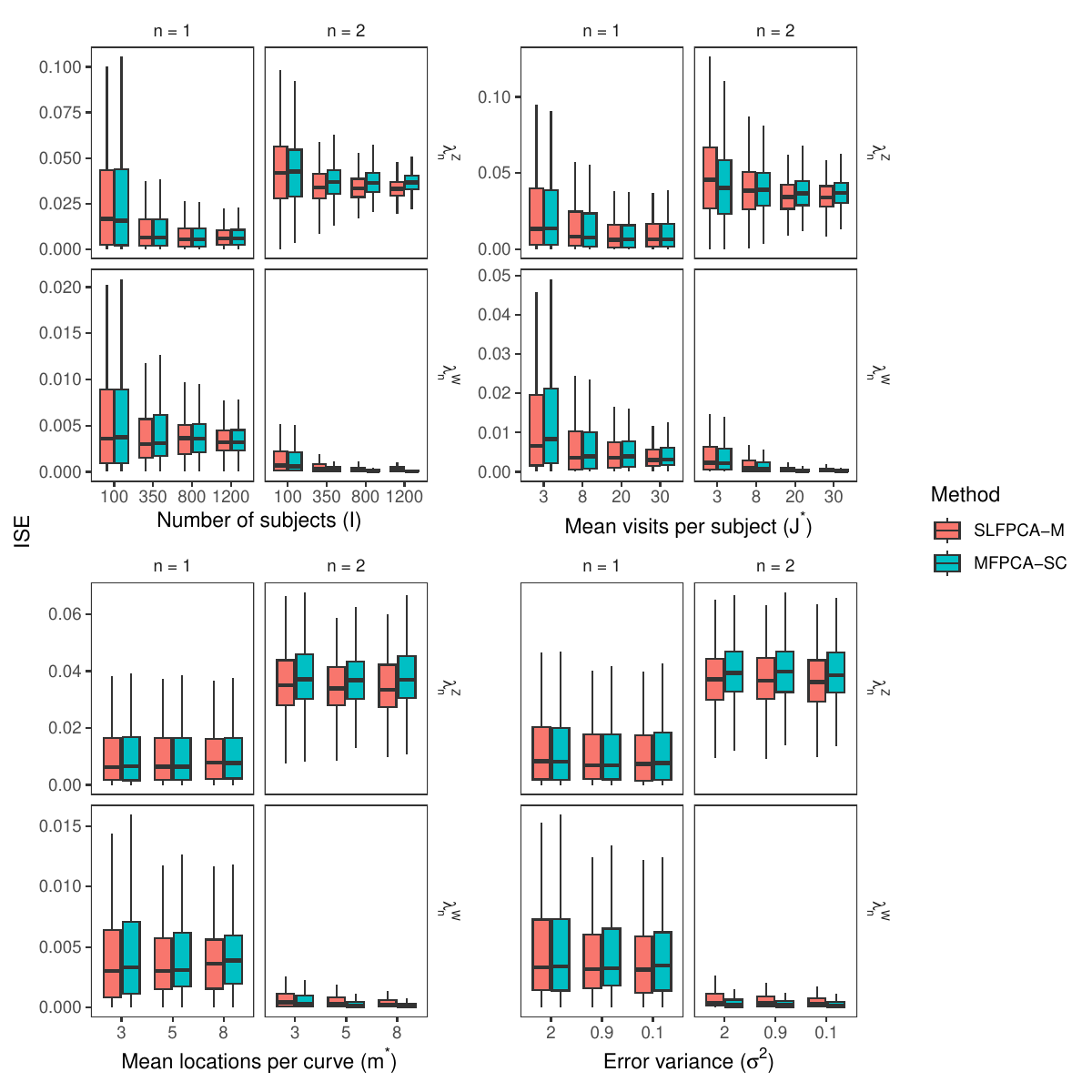}
    \caption{Boxplots of MSE of eigenvalues in Simulation 2 (MFPCA). 
    The base parameters are $I = 350$, $J^* = 30$, $m^* = 5$, and $\sigma^2 = .85$, and one parameter is varied at a time.
    The proposed methods and MFPCA \citep{di2009mfpca} perform comparably.}
    \label{fig:s_ev_mse_mfpca}
\end{figure}

\begin{figure}[H]
    \centering
    \includegraphics[width=0.625\linewidth]{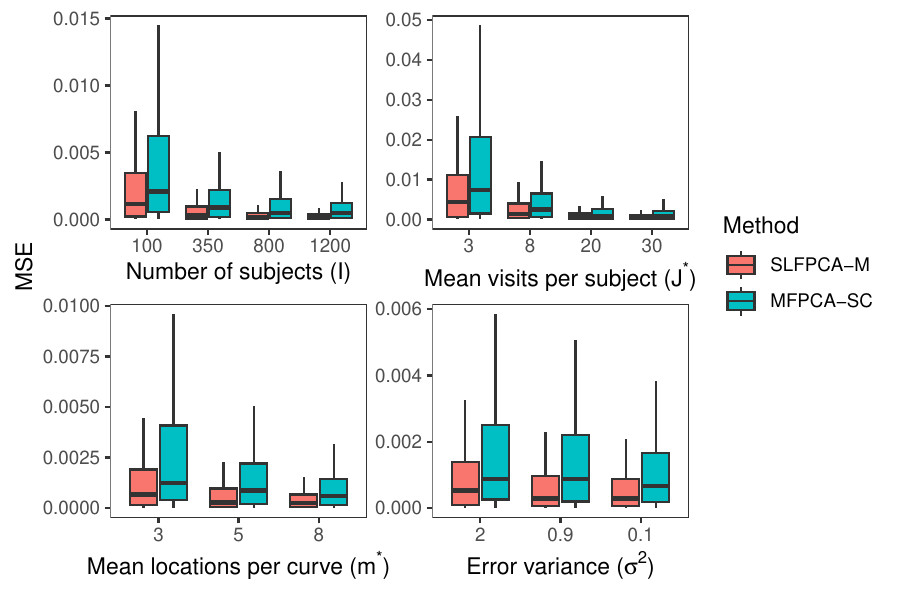}
    \caption{Boxplots of MSE of error variance estimation in Simulation 2 (MFPCA). 
    The base parameters are $I = 350$, $J^* = 30$, $m^* = 5$, and $\sigma^2 = .85$, and one parameter is varied at a time.}
    \label{fig:s_sigma2_mse_mfpca}
\end{figure}

Fast MFPCA results were excluded from the plots above because its substantially lower estimation accuracy would obscure the performance differences between MFPCA-SC and SLFPCA-M; the results with Fast MFPCA are given below.

\begin{figure}[H]
    \centering
    \includegraphics[width=\linewidth]{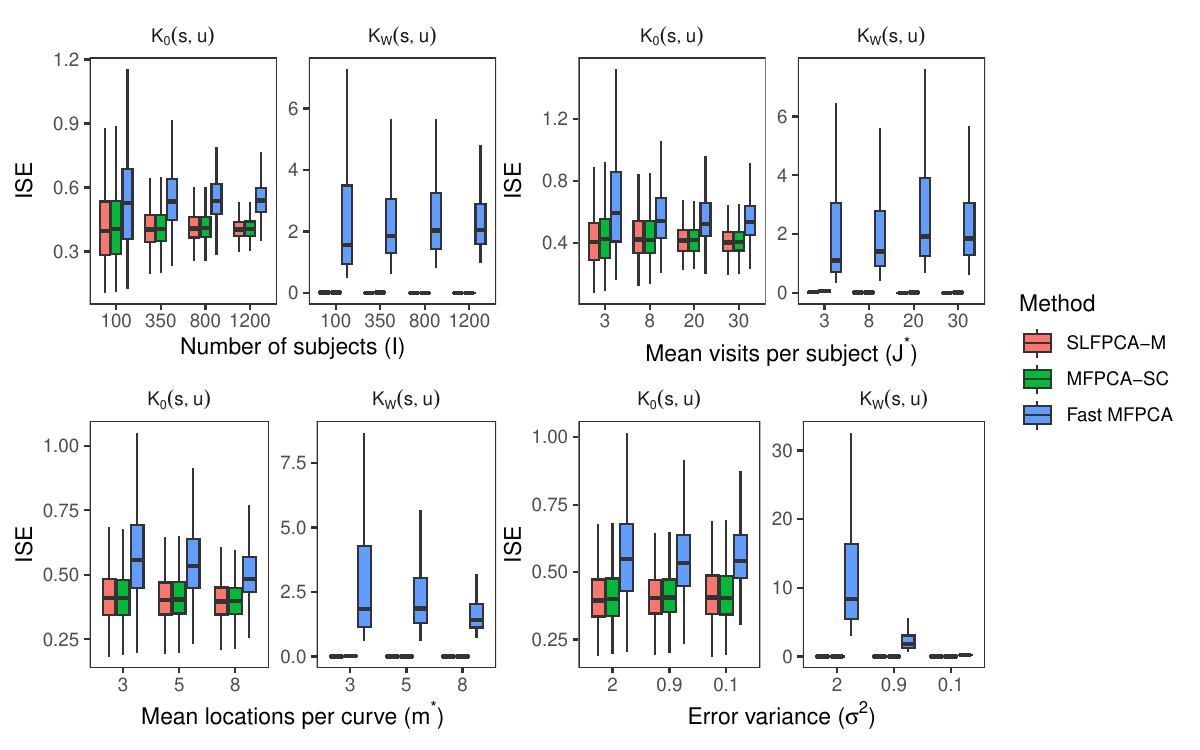}
    \caption{ISE of covariance functions in Simulation 2 (MFPCA), similar to Figure~\ref{fig:s_cov_ise_mfpca} but including Fast MFPCA.}
    \label{fig:s_cov_ise_mfpca_f}
\end{figure}

\begin{figure}[H]
    \centering
    \includegraphics[width=\linewidth]{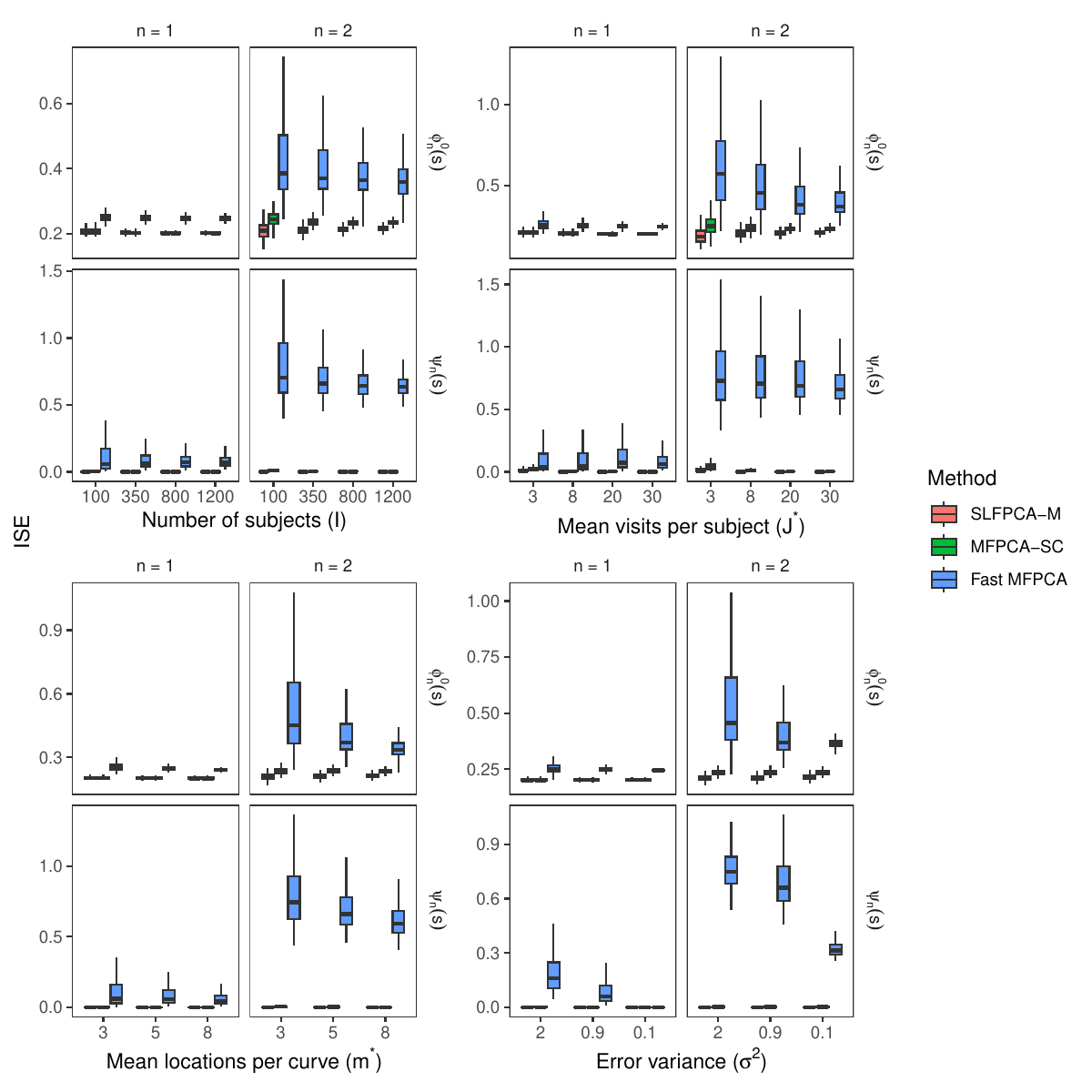}
    \caption{ISE of eigenfunctions in Simulation 2 (MFPCA), similar to Figure~\ref{fig:s_ef_ise_mfpca} but including Fast MFPCA.}
    \label{fig:s_ef_ise_mfpca_f}
\end{figure}

\begin{figure}[H]
    \centering
    \includegraphics[width=\linewidth]{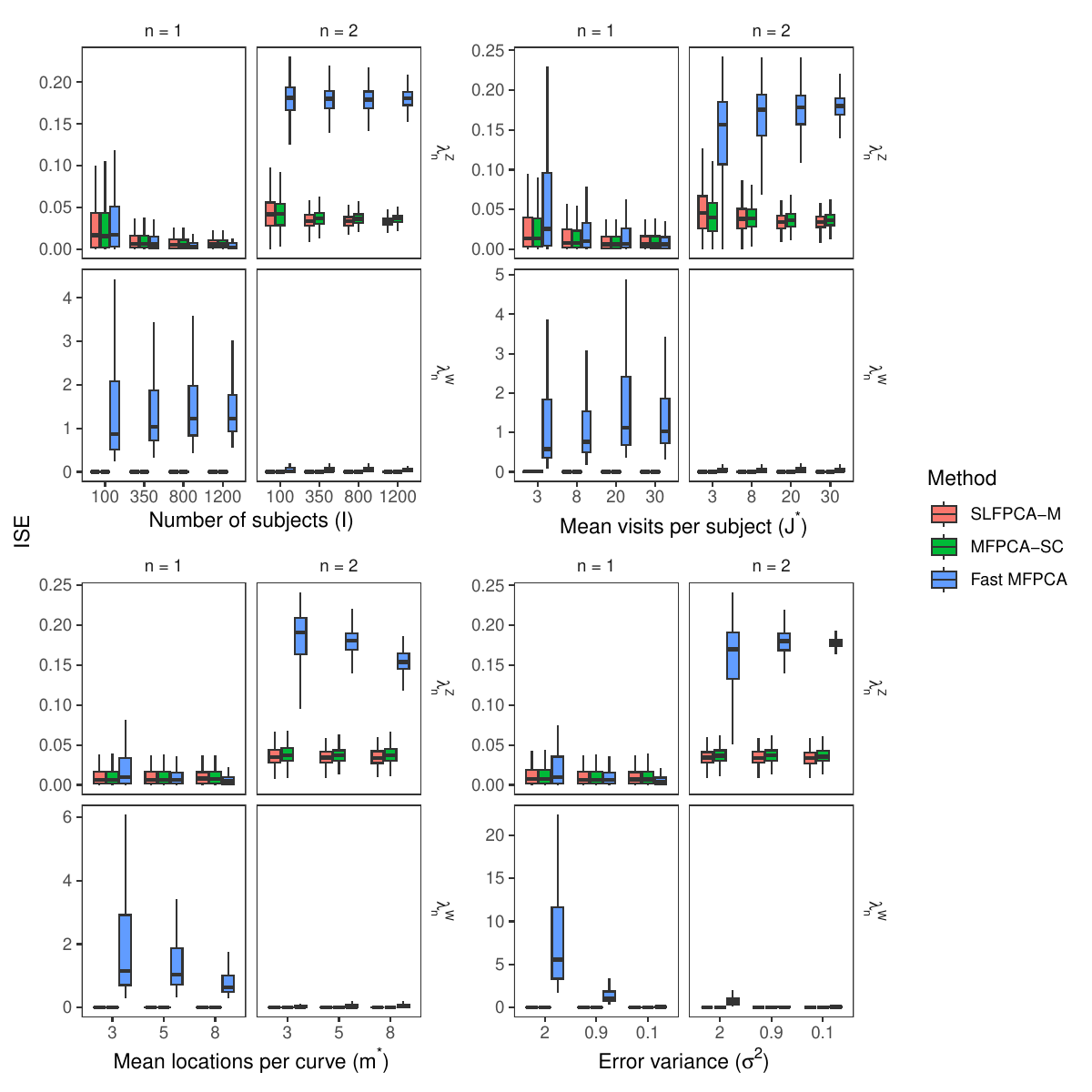}
    \caption{MSE of eigenvalues in Simulation 2 (MFPCA), similar to Figure~\ref{fig:s_ev_mse_mfpca} but including Fast MFPCA.}
    \label{fig:s_ev_mse_mfpca_f}
\end{figure}

\begin{figure}[H]
    \centering
    \includegraphics[width=0.625\linewidth]{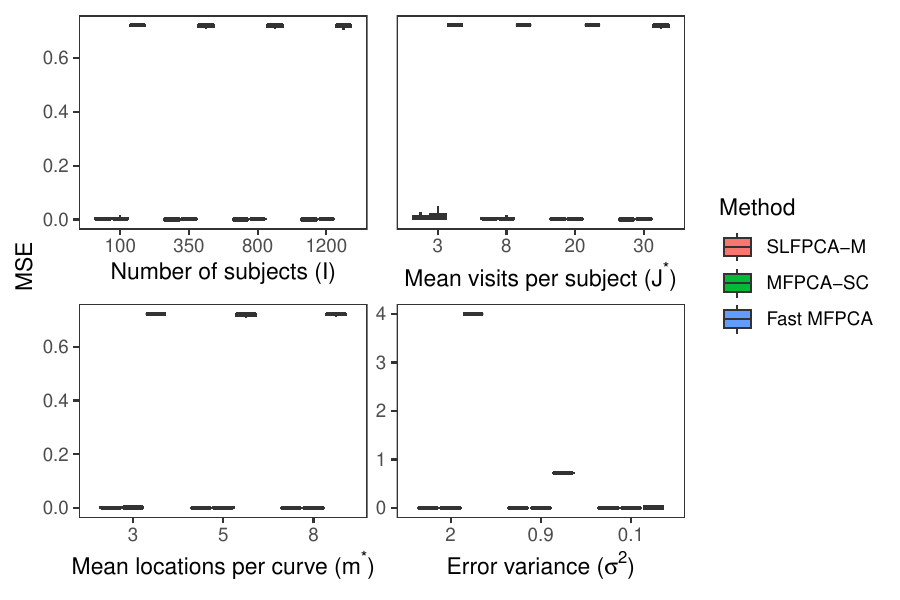}
    \caption{MSE of error variance in Simulation 2 (MFPCA), similar to Figure~\ref{fig:s_sigma2_mse_mfpca} but including Fast MFPCA.}
    \label{fig:s_sigma2_mse_mfpca_f}
\end{figure}

\begin{figure}[H]
    \centering
    \includegraphics[width=0.625\linewidth]{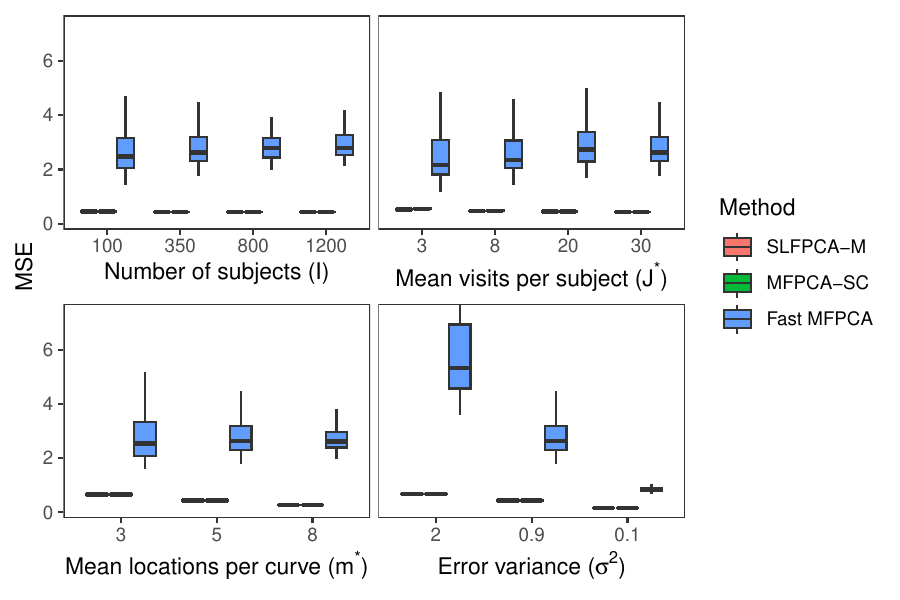}
    \caption{MSE of predictions in Simulation 2 (MFPCA), similar to Figure~\ref{fig:s_pred_mse_mfpca} but including Fast MFPCA.}
    \label{fig:s_pred_mse_mfpca_f}
\end{figure}

\subsection{Simulation 3: LFPCA with complete data}
\label{subsec:s_greven_sim}

Because the available implementation of LFPCA-G \citep{greven2010lfpca} does not accommodate missing data, we compare SLFPCA to LFPCA-G in the complete data scenario.

\subsubsection{Simulation design}
\label{subsubsec:s_greven_sim_design}

The simulation design is identical to the design presented in Section~3.1 of the main text, except that all curves $Y_{ij}(\cdot)$ are observed on a common grid $\mathcal{S}^*$ of timepoints rather than randomly sampled, curve-specific timepoints.
Each curve has $|\mathcal{S}^*| = 100$ points equally spaced on $[0, 1]$.
In addition, the parameters $I$ and $J$ are reduced for computation time. The base simulation parameters are $I = 80$, $J^* = 5$, and $\sigma^2 = .85$; the variations are $I = 40$, $J^* = 3$, and $\sigma^2 = .1, 2$. 

\subsubsection{Results}
\label{subsubsec:s_greven_sim_results}

First, the distribution of covariance function ISE is shown in Figure~\ref{fig:s_cov_ise_complete}.
Across all parameters, SLFPCA achieves better performance for $K_0(s,u)$, LFPCA-G performs better for $K_{01}(s,u)$ and $K_1(s,u)$, and estimation of $K_W(s,u)$ is comparable between the two methods.
Higher $I$ and $J$ lead to lower MSE for both methods, as expected, but decreasing error variance does not seem to impact the ISEs as much.

Figure~\ref{fig:s_ef_ise_complete} shows the eigenfunction estimation results. While the SLFPCA estimates $K_0(s,u)$ better, it only has an advantage in estimating the first eigenfunction $\phi^0_1(u)$; the second eigenfunction is estimated better by LFPCA-G.
For $K_1(s,u)$, SLFPCA estimates the first eigenfunction better, despite worse performance estimating $K_1(s,u)$ overall.
Both eigenfunctions are estimated similarly for $K_W(s,u)$.
Again, these conclusions are consistent across variations in parameters.
Figure~\ref{fig:s_ev_mse_complete} shows the eigenvalue estimation results. For the between-subject eigenvalues, SLFPCA estimates the second eigenvalue better, while LFPCA-G estimates the first eigenvalue better, which is the opposite trend of the eigenfunctions. For within-subject, the methods perform similarly.

Figure~\ref{fig:s_pred_mse_complete} shows that SLPFCA outperforms LFPCA-G in all scenarios for prediction. The results for SLFPCA follow the expected patterns when increasing $I$, $J$, and $\sigma^2$, but the predictions from LFPCA-G do not improve as much when increasing $I$ or $J$. 
It is not clear why the prediction accuracy is so different despite similar accuracy in covariance estimation.

\begin{figure}
    \centering
    \includegraphics[width=\linewidth]{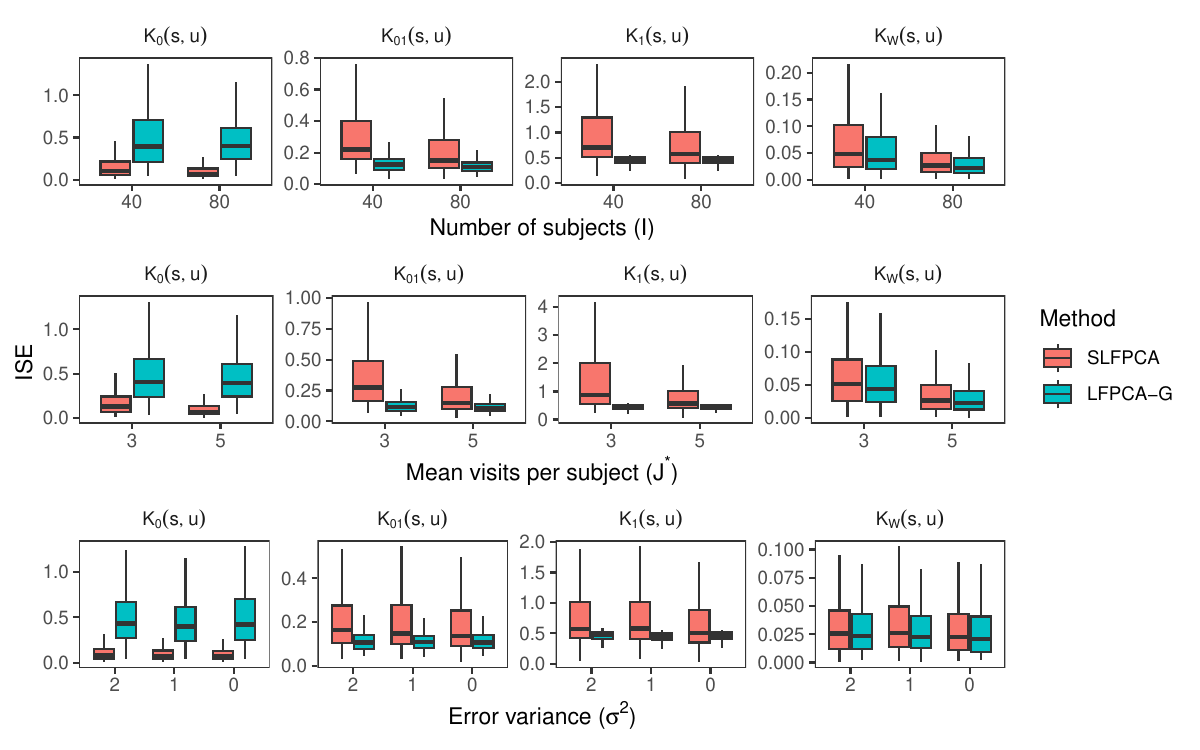}
    \caption{Boxplots of ISE of covariance functions in Simulation 3 (LFPCA with complete data). The base simulation parameters are $I = 80$, $J^* = 5$, and $\sigma^2 = .85$; the variations are $I = 40$, $J^* = 3$, and $\sigma^2 = .1, 2$. }
    \label{fig:s_cov_ise_complete}
\end{figure}

\begin{figure}
    \centering
    \includegraphics[width=\linewidth]{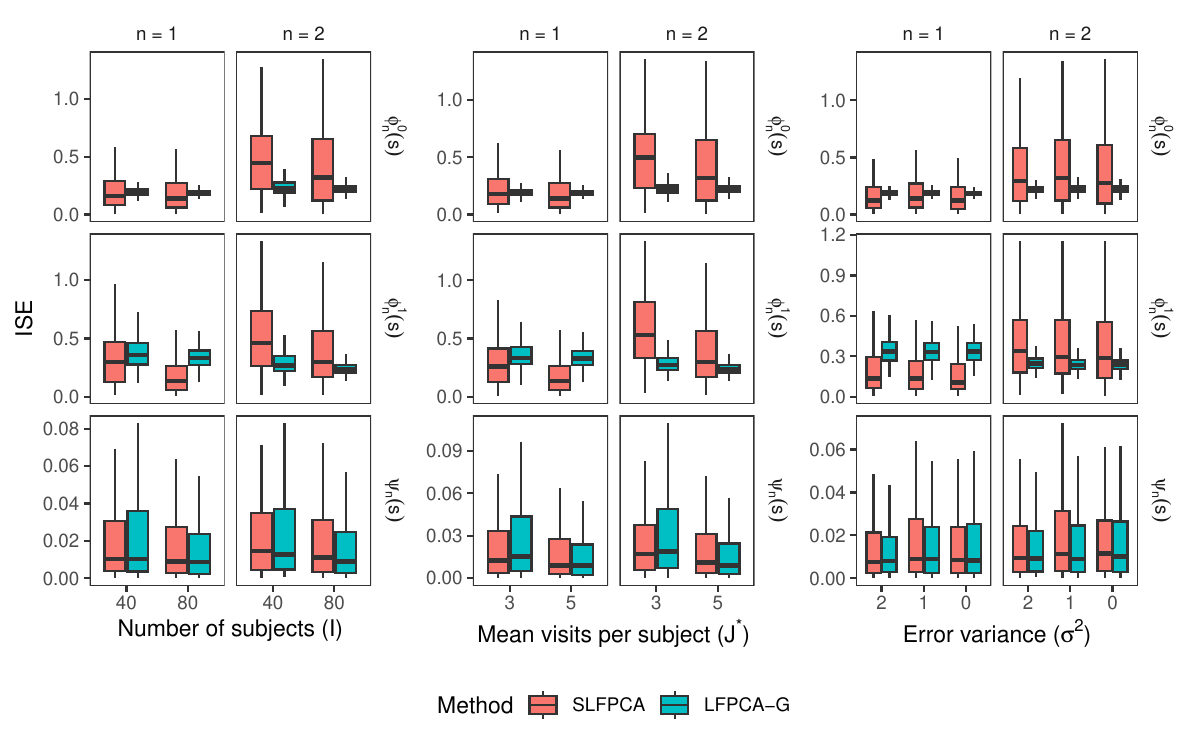}
    \caption{Boxplots of ISE of eigenfunctions in Simulation 3 (complete data). }
    \label{fig:s_ef_ise_complete}
\end{figure}

\begin{figure}
    \centering
    \includegraphics[width=\linewidth]{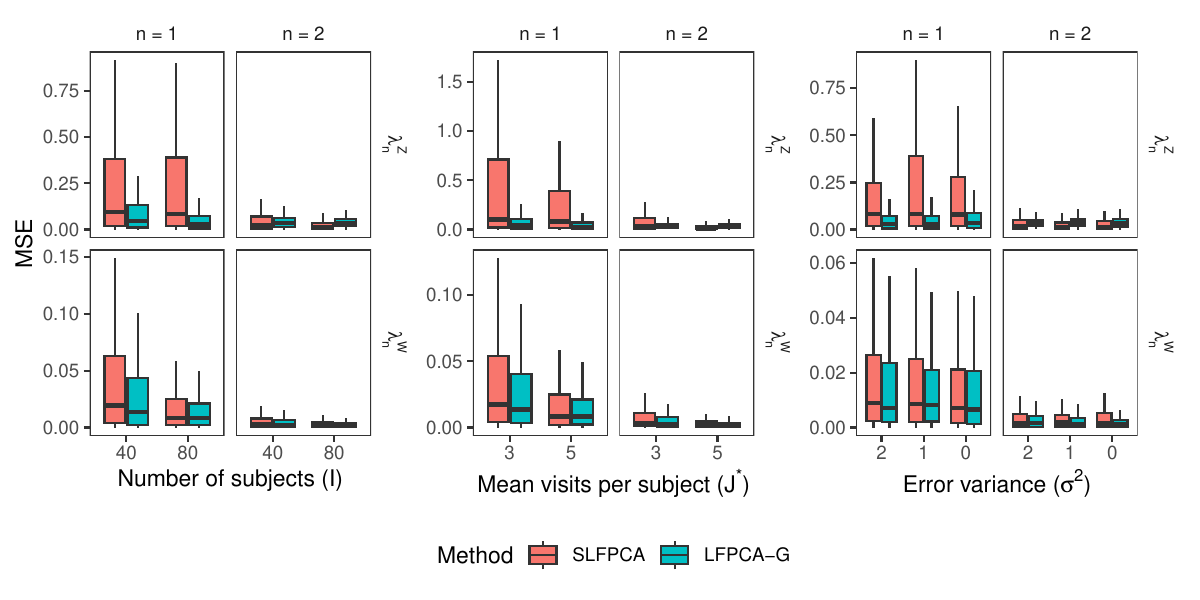}
    \caption{Boxplots of MSE of eigenvalues in Simulation 3 (complete data). }
    \label{fig:s_ev_mse_complete}
\end{figure}

\begin{figure}
    \centering
    \includegraphics[width=0.875\linewidth]{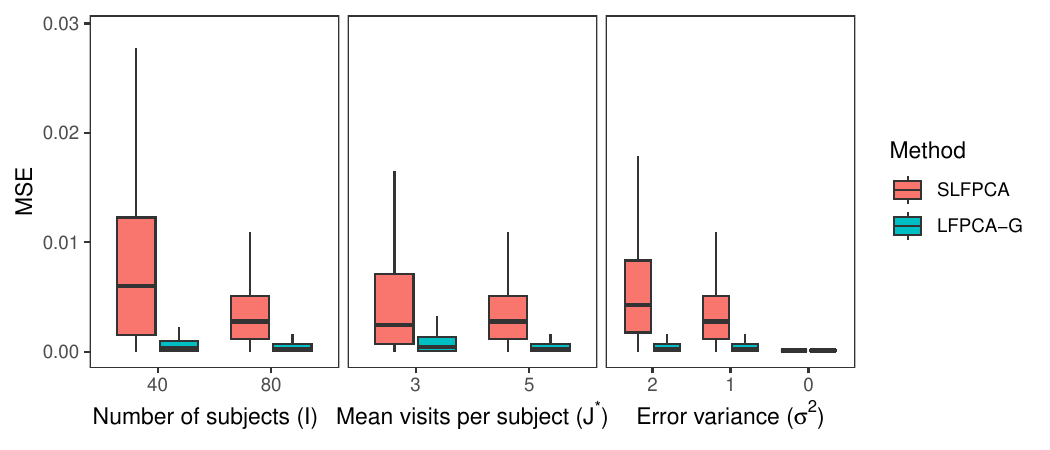}
    \caption{Boxplots of MSE of error variance in Simulation 3 (complete data). }
    \label{fig:s_sigma2_mse_complete}
\end{figure}

\begin{figure}
    \centering
    \includegraphics[width=.875\linewidth]{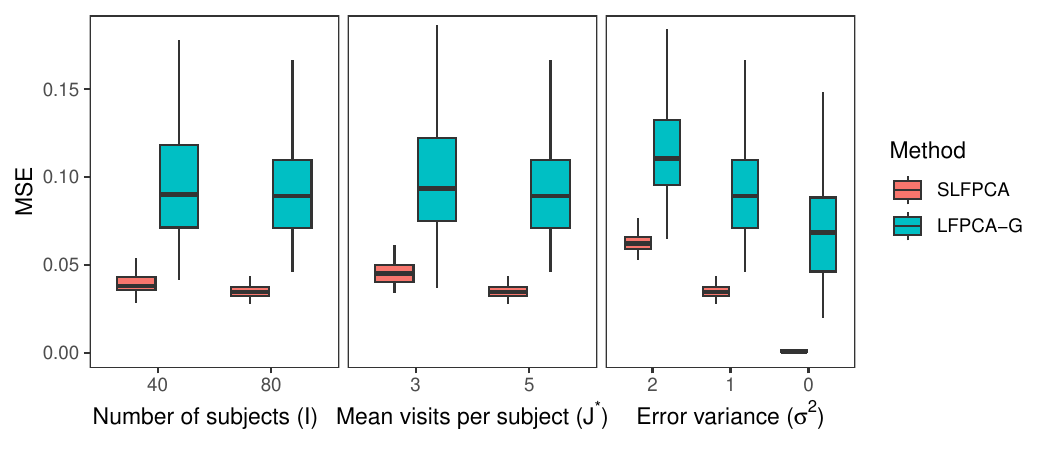}
    \caption{Boxplots of MSE of predictions in Simulation 3 (complete data). }
    \label{fig:s_pred_mse_complete}
\end{figure}

\pagebreak

\section{Additional application results}
\label{section:s_application}

\begin{figure}[H]
    \centering
    \includegraphics[width=\linewidth]{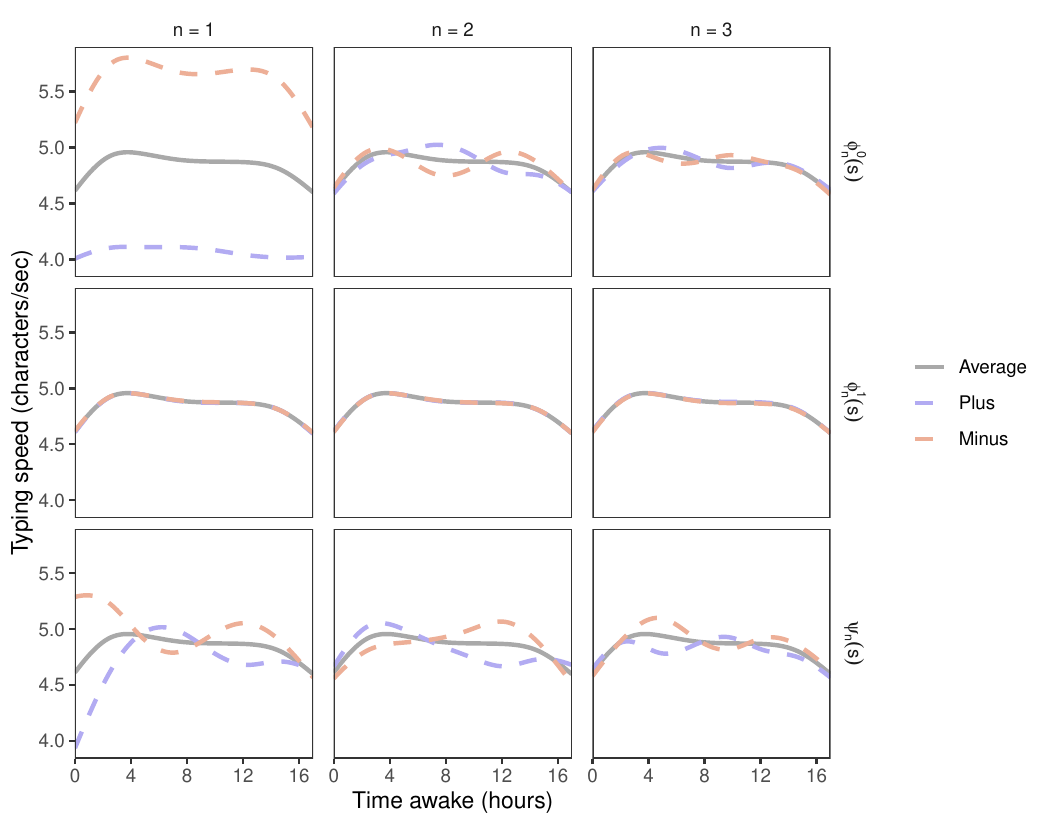}
    \caption{Marginal mean function plus and minus a suitable multiple of eigenfunctions in SensorKit typing speed data. Specifically, the three rows show $\mu(s) \pm \sqrt{\lambda_Z} \phi_n^0(s)$, $\mu(s) \pm \sqrt{\lambda_Z} \phi_n^1(s)$, and $\mu(s) \pm \sqrt{\lambda_W} \psi_n(s)$ for $n = 1, ..., 4$.}
    \label{fig:s_ef_pm_plot}
\end{figure}

\end{document}